\documentclass[journal,draftclsnofoot,onecolumn]{IEEEtran}
\pdfoutput=1
\usepackage{ifpdf}

\usepackage{amssymb}
\usepackage{}
\usepackage{mathrsfs,amsmath,amssymb,epsfig,amsfonts,fancyhdr,graphicx,cite,amsthm}
\usepackage{setspace}
\newtheorem{thm}{Theorem}[]
\newtheorem{cor}{Corollary}[]

\theoremstyle{remark}
\newtheorem{rem}[]{Remark}

\theoremstyle{definition}

\begin{document}

\title{On Compress-Forward without Wyner-Ziv Binning for Relay Networks}
\author{Peng Zhong, Ahmad Abu Al Haija, and Mai Vu
\thanks{The authors are with Department of Electrical and Computer Engineering, McGill University, Montreal, Canada
(e-mails: peng.zhong@mail.mcgill.ca; ahmad.abualhaija@mail.mcgill.ca; mai.h.vu@mcgill.ca).}
\thanks{A part of this paper was presented at the Allerton conference, Sept 2011.}
}

\maketitle

\begin{abstract}
Noisy network coding is recently proposed for the general multi-source
network by Lim, Kim, El Gamal and Chung. This scheme builds on
compress-forward (CF) relaying but involves three new ideas, namely no
Wyner-Ziv binning, relaxed simultaneous decoding and message
repetition. In this paper, using the two-way relay channel as the
underlining example, we analyze the impact of each of these
ideas on the achievable rate region of relay networks.\par

First, CF without binning but with joint decoding of both the message
and compression index can achieve a larger rate region than the
original CF scheme for multi-destination relay networks. With binning
and successive decoding, the compression rate at each relay is
constrained by the weakest link from the relay to a destination; but
without binning, this constraint is relaxed. Second, simultaneous
decoding of all messages over all blocks without uniquely decoding the
compression indices can remove the constraints on compression rate
completely, but is still subject to the message block boundary
effect. This boundary effect is caused by the decoding of signal
received in the last block, in which not all inputs and outputs may be
regularly involved. Third, message repetition is necessary to overcome
this boundary effect and achieve the noisy network coding region for
multi-source networks. The rate region is enlarged with increasing
repetition times.\par

We also apply CF without binning specifically to the one-way and
two-way relay channels and analyze the rate regions in detail. For the
one-way relay channel, it achieves the same rate as the original CF
and noisy network coding but has only 1 block decoding delay. For the
two-way relay channel, we derive the explicit channel conditions in
the Gaussian and fading cases for CF without binning to achieve the
same rate region or sum rate as noisy network  coding. These analyses
may be appealing to practical implementation because of the shorter
encoding and decoding delay in CF without binning.

\end{abstract}


\IEEEpeerreviewmaketitle
\section{Introduction}\label{sec:intro}

\IEEEPARstart{T}{he} relay channel (RC) first introduced by van der Meulen \cite{meulen1971three} consists of a source aiming to communicate with a destination with the help of a relay node. Relay channels can be extended to relay networks, in which each node wishes to send a message to some destinations while also acting as a relay for others. Several relay strategies including decode-forward, compress-forward and amplify-forward have been proposed for the relay channel. In this paper, we focus on compress-forward based strategies, analyze its variants and applications to the one-way, two-way relay channels and the general relay networks.
\subsection{Related work}
We first review related works on compress-forward (CF) strategies for relay channels and networks. We divide the discussion into three parts. The first part is on single-source, single-destination relay networks. The second part is on some variants of the CF scheme. The third part is on multi-source multi-destination relay networks. We attempt to make this review compressive but do not claim it to be exhaustive of existing works.

\subsubsection{Single-source single-destination relay networks}
In the following works, the source and relay encoding are similar. At each block, the source sends a different message; the relay first compresses its received signal then uses Wyner-Ziv binning to reduce the forwarding rate. The differences are mainly in the decoding at the destination by either performing successive or joint decoding.
\begin{itemize}
\item Compress-forward is originally proposed for the 3-node single-relay channel (also called the one-way relay channel) by Cover and El Gamal in \cite{cover1979capacity}. The source sends a new message at each block using independent codebooks. The relay compresses its noisy observation of the source signal and forwards the bin index of the compression to the destination using Wyner-Ziv coding \cite{Wyner1976the}. A 3-step sequential decoding is then performed at the destination. At the end of each block, the destination first decodes the bin index, then decodes the compression index within that bin, and at last uses this compression index to decode the message sent in the previous block. The following rate is achievable with the 3-step sequential decoding CF scheme:
    \begin{align}
    \label{3-step-decoding}
    R\leq I(X;Y,\hat{Y}_r|X_r)
    \end{align}
    subject to
    \begin{align*}
    I(X_r;Y)\geq I(\hat{Y}_r;Y_r|X_r,Y).
    \end{align*}
    for some $p(x)p(x_r)p(\hat{y}_r|y_r,x_r)p(y,y_r|x,x_r)$.
\item El Gamal, Mohseni, and Zahedi put forward a 2-step decoding CF scheme in \cite{gamal2006bounds}. The source and relay perform the same encoding as that in \cite{cover1979capacity}. The destination, however, decodes in 2 sequential steps.
    At the end of each block, it decodes the bin index first, then decodes the message for some compression indices within that bin instead of decoding the compression index precisely. With this 2-step decoding CF scheme, the following rate is achievable:
    \begin{align}
    R\leq \min \{I(X,X_r;Y)-I(\hat{Y}_r;Y_r|X,X_r,Y), I(X;Y,\hat{Y}_r|X_r)\}
    \end{align}
    for some $p(x)p(x_r)p(\hat{y}_r|y_r,x_r)p(y,y_r|x,x_r)$.
    It has been shown \cite{gamal2006bounds}\cite{gamal2010lecture} that this 2-step decoding CF achieves the same rate as the original 3-step decoding CF in \eqref{3-step-decoding} but has a simpler representation.
\item  Kramer, Gastpar, and Gupta extend the 3-step decoding CF scheme to the single-source, single-destination and multiple-relay network in \cite{Kramer2005cooperative}.  The relays can also cooperate with each other to transmit the compression bin indices by partially decoding these bin indices.
\item Chong, Motani and Garg propose two coding schemes for the one-way relay channel combining decode-forward and compress-forward in \cite{Chong2007generalized}. Similar to the original combined scheme in \cite{cover1979capacity}, the source splits its message into two parts and the relay decode-forwards one part and compress-forwards the other. The destination, however, performs backward decoding either successively or simultaneously.
    These two
    strategies achieve higher rates than the original combined strategy in \cite{cover1979capacity} for certain parameters of the Gaussian relay channel.
\end{itemize}

\subsubsection{Variants of compress-forward}
Several variants of the CF scheme have been proposed for the relay channel.
\begin{itemize}
\item Cover and Kim propose a hash-forward (HF) scheme for the deterministic relay channel in \cite{Cover2007capacity}, in which the relay hashes (randomly bins) its observation directly without compression and forwards the bin index to the destination. HF achieves the capacity of the deterministic relay channel.
    Kim then proposes an extended hash-forward (EHF) scheme in \cite{Kim2007coding} which allows the destination to perform list decoding of the source messages for the general non-deterministic case.
\item  Razaghi and Yu introduce in \cite{Razaghi2010universal} a generalized hash-forward (GHF) relay strategy which allows the relay to choose a description of a general form rather than direct hashing (binning) of its received signal, but with a description rate on the opposite regime of Wyner-Ziv binning. The destination then performs list decoding of the description indices. GHF achieves the same rate as the original CF for the one-way relay channel but have been shown to exhibit advantage for multi-destination networks by allowing different description rates to different destinations \cite{Zhou2011incremental}.

\item Recently a new notion of quantize-forward or CF without binning emerges \cite{sung2011noisy} \cite{Kramer2011short} in which the relay compresses its received signal but forwards the compression index directly without using Wyner-Ziv binning. We discuss this idea in more details in the next few paragraphs.
\end{itemize}

\subsubsection{Multi-source multi-destination relay networks}
Relatively fewer works have applied CF to the general multi-source multi-destination relay network.
\begin{itemize}
\item Rankov and Wittneben applied the 3-step decoding CF scheme to the two-way relay channel (TWRC) in \cite{rankov2006achievable}, in which two users wish to exchange messages with the help of a relay. The encoding and decoding are similar to those in \cite{cover1979capacity}.
\item Recently, Lim, Kim, El Gamal and Chung put forward a noisy network coding scheme \cite{sung2011noisy} for the general multi-source noisy network. This scheme involves three key new ideas. The first is message repetition, in which the same message is sent multiple times over consecutive blocks using independent codebooks. Second, each relay does not use Wyner-Ziv binning but only compresses its received signal and forwards the compression index directly. Third, each destination performs simultaneous decoding of the message based on signals received from all blocks without uniquely decoding the compression indices.
    Noisy network coding simplifies to the capacity-achieving network coding for the noiseless multicast network. Compared to the original CF, it achieves the same rate for the one-way relay channel and achieves a larger rate region when applied to multi-source networks such as the two-way relay channel. However, it also brings more delay in decoding because of message repetition.
\end{itemize}

\subsubsection{Analysis of compress-forward schemes}
With the above variants and developments on CF relaying, some works have analyzed the different ideas in compress-forward.
\begin{itemize}
\item Kim, Skoglund and Caire \cite{Kim2009quantifying} show that without Wyner-Ziv binning at the relay, using sequential decoding at the destination incurs rate loss in the one-way relay channel. The amount of rate loss is quantified specifically in terms of the diversity-multiplexing tradeoff for the fading channel.
\item Wu and Xie demonstrate in \cite{Wuonallerton2010} that for single-source, single-destination and multiple-relay networks, using the original CF encoding with Wyner-Ziv binning of \cite{cover1979capacity}, there is no improvement on the achievable rate by joint decoding of the message and compression indices. To maximize the CF achievable rate, the compression rate should always be
    chosen to support successive decoding.

\item Wu and Xie then propose in \cite{Wu2010on} for the single-source, single-destination and multiple-relay network a scheme that achieves the same rate as noisy network coding \cite{sung2011noisy} but with the simpler classical encoding of \cite{cover1979capacity} and backward decoding. The backward decoding involves first decoding the compression indices then successively decoding the messages backward. It requires, however, extending the relay forwarding times for a number of blocks without sending new messages, which causes an albeit small but non-vanishing rate loss.
\item Kramer and Hou discuss in \cite{Kramer2011short} a short-message quantize-forward scheme without message repetition or Wyner-Ziv binning but with joint decoding of the message and compression index at the destination. It also achieves the same rate as the original CF and noisy network coding for the one-way relay channel.
This result is similar to our result in Section \ref{sec:one-way}; however, they are independent results and we were only aware of the work in \cite{Kramer2011short} after publishing our conference paper \cite{zhong2011compress}.
\end{itemize}

\subsection{Summary of main results}
In this paper, we analyze the impact of each of the 3 new ideas in noisy network coding, namely no Wyner-Ziv binning, simultaneous decoding and message repetition. We apply these ideas separately to the one-way, two-way relay channels and extend the analysis to the multi-source multi-destination relay network.
\begin{itemize}
\item For the one-way relay channel, we derive the achievable rate using CF without binning (also called quantize-forward) but with joint decoding of both the message and compression index. It achieves the same rate as the original CF in \cite{cover1979capacity} \cite{gamal2006bounds}. Compared with the original CF, it simplifies relay operation since Wyner-Ziv binning is not needed, but increases decoding complexity at the destination since joint decoding instead of successive decoding is required. Compared with noisy network coding, it achieves the same rate while having much less encoding and decoding delay.
    This part is similar to Kramer and Hou's result in \cite{Kramer2011short}, but is derived independently.
\item We next extend CF without binning to the two-way relay channel and derive its achievable rate region. The scheme achieves a larger rate region than the original CF \cite{rankov2006achievable}. With binning and successive decoding, the compression rate is constrained by the weaker of the links from relay to two users. But without binning, this constraint is relaxed. However, CF without binning generally achieves smaller rate region than noisy network coding \cite{sung2011noisy}. In CF without binning, the decoding of the compression index imposes constraints on the compression rate. In noisy network coding, the destinations do not decode the compression index explicitly, thus removing these constraints.
\item Using the two-way relay channel as the underlining example, we then analyze the effect of each of the three new ideas in noisy network coding for the general multi-source multi-destination relay networks. The main findings are:
    \begin{itemize}
\item No Wyner-Ziv binning requires joint decoding of both the message and compression index at the destination for all networks.
\item Comparing with the original CF \cite{cover1979capacity},
     no binning but with joint decoding can enlarge the rate region for multi-source multi-destination relay networks. This finding is contrary to the case of single-source single-destination relay networks in \cite{Wu2010on}.
\item CF without binning generally achieves a smaller rate region than noisy network coding because of constraints on compression rates.
    Simultaneous decoding over all blocks without explicitly decoding compression indices does remove this constraint but is still not enough. This is due to the block boundary effect when sending a different message at each block.
\item Message repetition is necessary to achieve the noisy network coding rate region for multi-source \mbox{networks.} It helps mitigate the block boundary effect as repetition times increase. To completely eliminate the boundary effect as in noisy network coding, however, the repetition times need to approach infinity.
    \end{itemize}
\item We also provide detailed analysis of the achievable rate regions for the Gaussian and fading two-way relay channels. Comparing the achievable rate region and the sum rate of the original CF, CF without binning and noisy network coding, we derive explicit conditions for which CF without binning achieves the same rate region or sum rate as the original CF or noisy network coding. These analyses may have implication in practical implementation of CF without binning in these channels.

\end{itemize}

\section{One-Way Relay Channel}\label{sec:one-way}
\subsection{Discrete memoryless one-way relay channel model}
The discrete memoryless one-way relay channel (DM-RC) is
denoted by $(\mathcal{X} \times \mathcal{X}_r ,p(y,y_r|x,x_r),\mathcal{Y} \times
\mathcal{Y}_r)$, as in Figure \ref{fig:one-way}. Sender $X$ wishes to send a message $M$ to receiver $Y$ with the help of the relay $(X_r,Y_r)$. We consider a full-duplex
channel in which all nodes can transmit and receive at the same time.

\begin{figure}[!t]
\centering
  \includegraphics[scale=0.7]{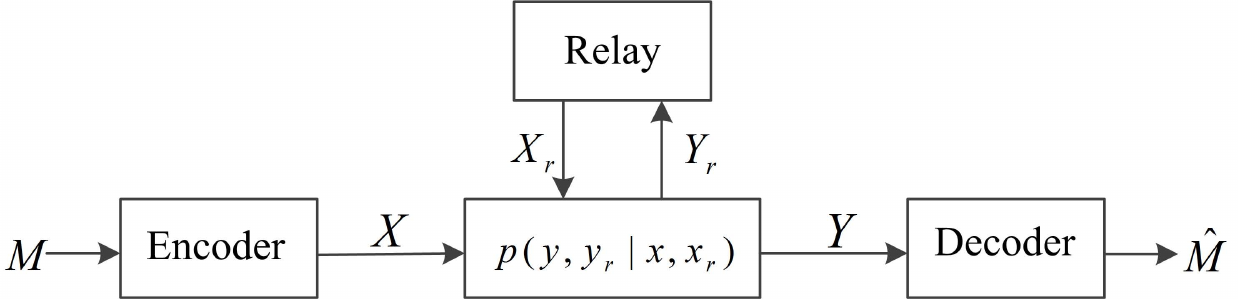}\\
  \caption{One-way relay channel model.}\label{fig:one-way}
\end{figure}

A $(2^{nR},n,P_e)$ code for a DM-RC consists of: a
message set $\mathcal{M}=[1:2^{nR}]$; an encoder that assigns a codeword $x^n(m)$ to each message $m\in [1:2^{nR}]$; a relay encoder that assigns at time $i\in [1:n]$ a symbol $x_{ri}(y_r^{i-1})$ to each past received sequence $y_r^{i-1}\in \mathcal{Y}_r^{i-1}$; a decoder that assigns a message $\hat{m}$ or an error message to each received sequence $y^{n}\in \mathcal{Y}^{n}$. The average error probability is
$P_e=\textrm{Pr}\{\hat{M}\neq M\}$. The rate $R$ is said to be achievable for the DM-RC if there exists a sequence of $(2^{nR},n,P_e)$ codes with $P_e\rightarrow 0$. The supremum of all achievable rates is the capacity of the DM-RC.\par

\begin{table}[!t]
\centering
\begin{spacing}{1.5}
\begin{tabular}{|c|c|c|c|c|}
\hline
Block&\ldots&$j$&$j+1$&\ldots  \\
\hline
$X$ & \ldots & $x^n(m_j)$  & $x^n(m_{j+1})$  &\ldots \\
\hline
$Y_r$ & \ldots & $\hat{y}_r(k_j|k_{j-1})$ & $\hat{y}_r(k_{j+1}|k_j)$ &\ldots \\
\hline
$X_r$ & \ldots& $x^n_r({k_{j-1}})$ & $x^n_r({k_j})$ &\ldots \\
\hline
$Y$ & \ldots & $\hat{k}_{j-1},\hat{m}_{j-1}$ &$\hat{k}_j,\hat{m}_j$ &\ldots \\
\hline
\end{tabular}
\end{spacing}
\caption{\textnormal{Encoding and decoding of CF without binning for the one-way relay channel.}}
\label{tab:CF_one-way}
\end{table}

\subsection{Achievable rate for compress-forward without binning}
In the original CF scheme \cite{cover1979capacity} \cite{gamal2006bounds}, the source sends a new message in each block. The relay forwards the bin index of the description of its received signal. The receiver uses successive decoding to decode the bin index first, then decode the message from the sender. Here we analyze a CF scheme in which the relay forwards the index of the description of its received signal directly without binning while the receiver jointly decodes the index and message at the same time. We show that CF without binning can achieve same rate as the original CF scheme with binning.\par
The encoding and decoding of CF without binning are as follows (also see Table \ref{tab:CF_one-way}).
We use a block coding scheme in which each user sends $b-1$ messages over $b$ blocks of $n$ symbols each.\par
\subsubsection{Codebook generation}
Fix $p(x)p(x_r)p(\hat{y}_r|y_r,x_r)$. We randomly and independently generate a codebook for each block $j\in [1:b]$ as follows.
\begin{itemize}
\item Independently generate $2^{nR}$ sequences $x^n(m_j)\sim\prod ^n_{i=1}p(x_i) $, where $m_j \in [1:2^{nR}]$.
\item Independently generate $2^{nR_r}$ sequences $x_r^n(k_{j-1})\sim\prod ^n_{i=1}p(x_{ri})$, where $k_{j-1} \in [1:2^{nR_r}]$.
\item For each $k_{j-1}\in [1:2^{nR_r}]$, independently generate $2^{nR_r}$ sequences $\hat{y}^n_r(k_j|k_{j-1})\sim\prod ^n_{i=1}p(\hat{y}_{ri}|x_{ri}(k_{j-1}))$ where $k_{j}\in [1:2^{nR_r}]$.
\end{itemize}

\subsubsection{Encoding}
The source transmits $x^n(m_j)$ in block $j$. The relay, upon receiving $y^n_r(j)$, finds an index $k_j$ such that
\begin{align*}
(\hat{y}_r^n(k_j|k_{j-1}),y^n_r(j),x^n_r(k_{j-1}))\in T^{(n)}_{\epsilon},
\end{align*}
where $T^{(n)}_{\epsilon}$ denote the strong $\epsilon$-typical set \cite{gamal2010lecture}.
 Assume that such $k_j$ is found, the relay sends $x^n_r(k_j)$ in block $j+1$.\par

\subsubsection{Decoding}
 Assume the receiver has decoded $k_{j-1}$ correctly in block $j$. Then in block $j+1$, the receiver finds a unique pair of $(\hat{m}_j,\hat{k}_j)$ such that
\begin{align}
(x_r^n(\hat{k}_j),y^n(j+1))&\in T^{(n)}_{\epsilon} \nonumber \\
\textrm{and}~~~~(x^n(\hat{m}_j),x_r^n(\hat{k}_{j-1}),\hat{y}_r^n(\hat{k}_j|\hat{k}_{j-1}),y^n(j))&\in T^{(n)}_{\epsilon}. \nonumber
\end{align}

\begin{thm}
\label{thm_rate}
Consider a compress-forward scheme in which the relay does not use Wyner-Ziv binning but sends the compression index directly and the destination performs joint decoding of both the message and compression index.
The following rate is achievable for the one-way relay channel:
\begin{align}
R\leq \min \{I(X,X_r;Y)-I(\hat{Y}_r;Y_r|X,X_r,Y), I(X;Y,\hat{Y}_r|X_r)\}
\end{align}
subject to
\begin{align}
\label{oneway_const_1}
I(X_r;Y)+I(\hat{Y}_r;X,Y|X_r)\geq I(\hat{Y}_r;Y_r|X_r)
\end{align}
for some $p(x)p(x_r)p(\hat{y}_r|y_r,x_r)p(y,y_r|x,x_r)$.
\end{thm}

\begin{proof}
See Appendix \ref{appendix:thm_rate}.
\end{proof}

\subsection{Comparison with the original compress-forward scheme}
\begin{thm}
\label{thm:oneway_ori}
Compress-forward without binning in Theorem \ref{thm_rate} achieves the same rate as the original compress-forward scheme for the one-way relay channel, which is:
\begin{align}
R\leq \min \{I(X,X_r;Y)-I(\hat{Y}_r;Y_r|X,X_r,Y), I(X;Y,\hat{Y}_r|X_r)\}
\end{align}
for some $p(x)p(x_r)p(\hat{y}_r|y_r,x_r)p(y,y_r|x,x_r)$.
\end{thm}

\begin{proof}
To show that the rate region in Theorem \ref{thm_rate} is the same as the rate region in Theorem \ref{thm:oneway_ori}, we need to show that constraint (\ref{oneway_const_1}) is redundant. Note that an equivalent characterization of the rate region in Theorem \ref{thm:oneway_ori} is as follows \cite{cover1979capacity} \cite{gamal2006bounds} \cite{gamal2010lecture}:
\begin{align}
R\leq I(X;Y,\hat{Y}_r|X_r)
\end{align}
subject to
\begin{align}
\label{oneway_const_2}
I(X_r;Y)\geq I(\hat{Y}_r;Y_r|X_r,Y)
\end{align}
for some $p(x)p(x_r)p(\hat{y}_r|y_r,x_r)$.
Therefore, comparing (\ref{oneway_const_1}) with (\ref{oneway_const_2}), we only need to show that
\begin{align*}
I(\hat{Y}_r;Y_r|X_r,Y)\geq I(\hat{Y}_r;Y_r|X_r)-I(\hat{Y}_r;X,Y|X_r).
\end{align*}
This is true since
\begin{align*}
I(\hat{Y}_r;Y_r|X_r,Y)&=I(\hat{Y}_r;Y_r,X|X_r,Y)  \\
&=I(X;\hat{Y}_r|X_r,Y)+I(Y_r;\hat{Y}_r|X,X_r,Y)  \\
&\geq I(Y_r;\hat{Y}_r|X,X_r,Y)  \\
&=I(\hat{Y}_r;X,Y,Y_r|X_r)-I(\hat{Y}_r;X,Y|X_r)  \\
&=I(\hat{Y}_r;Y_r|X_r)-I(\hat{Y}_r;X,Y|X_r).
\end{align*}
\end{proof}

\begin{rem}
If using successive decoding, the rate achieved by CF without binning is strictly less than that with binning. Thus joint decoding is crucial for CF without binning.
\end{rem}

\begin{rem}
Joint decoding does not help improve the rate of the original CF with binning.
\end{rem}

\begin{rem}
Binning technique plays a role of allowing successive decoding instead of joint decoding, thus reduces destination decoding complexity. However, it has no impact on the achievable rate for the one-way relay channel. This effect on decoding complexity is similar to that in decode-forward, in which binning allows successive decoding \cite{cover1979capacity} while no binning requires backward decoding \cite{willems1985the}.
\end{rem}

\begin{rem}
For the one-way relay channel, CF without binning achieves the same rate region as GHF \cite{Razaghi2010universal}. However, GHF differs from CF without binning in that the relay still performs binning. In the decoding of GHF, the destination first decode the compression indices into a list, then use this list to help decode the source message. However, in CF without binning, the receiver performs joint decoding of the compression index and message.
\end{rem}

\begin{rem}
The result in Theorem \ref{thm_rate} and Theorem \ref{thm:oneway_ori}
is similar to that in \cite{Kramer2011short} by Kramer and Hou, but is derived independently; we only realize the work in \cite{Kramer2011short} after publishing our conference paper \cite{zhong2011compress}.
\end{rem}

\begin{rem}
An obvious benefit of this short message CF without binning scheme compared to noisy network coding is short encoding and decoding delay. A few other benefits such as low modulation complexity and potential MIMO gain are also recognized in \cite{Kramer2011short}.
\end{rem}

\begin{rem}
For the one-way relay channel, all the following schemes achieve the same rate: the original CF, CF without binning (QF), GHF, and noisy network coding. The difference in achievable rate only appears when applying to a multi-source and multi-destination network.
\end{rem}

\section{Two-way Relay Channel}\label{sec:two-way}
\subsection{Discrete memoryless two-way relay channel model}
The discrete memoryless two-way relay channel (DM-TWRC) is
denoted by $(\mathcal{X}_1 \times \mathcal{X}_2 \times
\mathcal{X}_r,p(y_1,y_2,y_r|x_1,x_2,x_r),\mathcal{Y}_1 \times
\mathcal{Y}_2 \times \mathcal{Y}_r)$, as in Figure \ref{p2}. Here
$x_1$ and $y_1$ are the input and output signals of user 1; $x_2$ and
$y_2$ are the input and output signals of user 2; $x_r$ and $y_r$ are
the input and output signals of the relay. We consider a full-duplex
channel in which all nodes can transmit and receive at the same time.

\begin{figure}[!t]
\centering
  \includegraphics[scale=0.7]{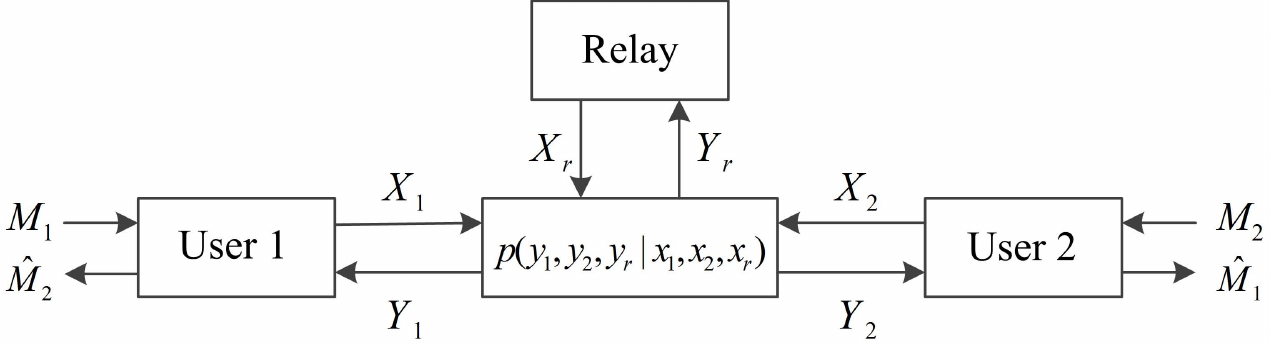}\\
  \caption{Two-way relay channel model.}\label{p2}
\end{figure}

A $(2^{nR_1},2^{nR_2},n,P_e)$ code for a DM-TWRC consists of two
message sets $\mathcal{M}_1=[1:2^{nR_1}]$ and
$\mathcal{M}_2=[1:2^{nR_2}]$, three encoding functions
$f_{1,i},f_{2,i},f_{r,i}$, $i=1, \ldots, n$ and two decoding function
$g_1,g_2$ as follows.
\begin{align}
x_{1,i}&=f_{1,i}(M_1,Y_{1,1},\ldots ,Y_{1,i-1}),~~~~i=1, \ldots, n\nonumber\\
x_{2,i}&=f_{2,i}(M_2,Y_{2,1},\ldots ,Y_{2,i-1}),~~~~i=1, \ldots, n\nonumber\\
x_{r,i}&=f_{r,i}(Y_{r,1},\ldots ,Y_{r,i-1}),~~~~i=1, \ldots, n\nonumber\\
g_1&:\mathcal{Y}^n_1 \times \mathcal{M}_1 \rightarrow \mathcal{M}_2\nonumber\\
g_2&:\mathcal{Y}^n_2 \times \mathcal{M}_2 \rightarrow \mathcal{M}_1\nonumber
\end{align}
The average error probability is
$P_e=\textrm{Pr}\{g_1(M_1,Y^n_1)\neq
M_2~\textrm{or}~g_2(M_2,Y^n_2)\neq M_1\}$. A rate pair is said to be
achievable if there exists a $(2^{nR_1},2^{nR_2},n,P_e)$ code such
that $P_e\rightarrow 0$ as $n\rightarrow \infty$. The closure of the
set of all achievable rates $(R_1,R_2)$ is the capacity region of the
two-way relay channel.\par

\subsection{Achievable rate region for compress-forward without binning}\label{sec:TWRC_CF_withoutbinning}
In this section, we extend CF without Wyner-Ziv binning but with joint decoding of both the message and compression index to the two-way relay channel. We then compare the achievable rate region with those by the original CF scheme and noisy network coding. Compared with the original CF \cite{rankov2006achievable}, CF without binning achieves a strictly larger rate region when the channel is asymmetric for the two users. Binning and successive decoding constrains the compression rate to the weaker of the channels from relay to two users. But without binning, this constraint is relaxed. Compared with noisy network coding, CF without binning achieves a smaller rate region. In CF without binning, the users need to decode the compression index precisely, which brings an extra constraint on the compression rate. However, this precise decoding is not necessary in noisy network coding.\par

Before presenting the achievable rate region of CF without binning, we outline its encoding and decoding techniques as follows. In each block, each user sends a new message using an independently generated codebook. At the end of each block, the relay finds a  description of its received signal from both users. Then it sends the codeword for the description index at the next block (instead of partitioning the description index into bins and sending the codeword for the bin index as in the original CF scheme). Each user jointly decodes the description index and message from the other user based on signals received in both the current and previous blocks. This decoding technique is different from that in the original decode-forward scheme, in which each user first decodes the bin index from the relay, and then decodes the message from the other user.\par

\begin{table}[!t]
\centering
\begin{spacing}{1.5}
\begin{tabular}{|c|c|c|c|c|}
\hline
Block&\ldots&$j$&$j+1$&\ldots  \\
\hline
$X_1$ & \ldots & $x^n_1(m_{1,j})$  & $x^n_1(m_{1,j+1})$  &\ldots \\
\hline
$X_2$ & \ldots& $x^n_2(m_{2,j})$  & $x^n_2(m_{2,j+1})$  &\ldots \\
\hline
$Y_r$ & \ldots & $\hat{y}_r(k_j|k_{j-1})$ & $\hat{y}_r(k_{j+1}|k_j)$ &\ldots \\
\hline
$X_r$ & \ldots& $x^n_r({k_{j-1}})$ & $x^n_r({k_{j}})$ &\ldots \\
\hline
$Y_1$ & \ldots & $\hat{k}_{j-1},\hat{m}_{2,j-1}$ &$\hat{k}_j,\hat{m}_{2,j}$ &\ldots \\
\hline
$Y_2$ & \ldots & $\hat{k}_{j-1},\hat{m}_{1,j-1}$ &$\hat{k}_j,\hat{m}_{1,j}$ &\ldots \\
\hline
\end{tabular}
\end{spacing}
\caption{\textnormal{Encoding and decoding of CF without binning for the two-way relay channel.}}
\label{tab:CF_two-way}
\end{table}

Specifically, we use a block coding scheme in which each user sends $b-1$ messages over $b$ blocks of $n$ symbols each (also see Table \ref{tab:CF_two-way}).\par
\subsubsection{Codebook generation}
 Fix joint distribution $p(x_1)p(x_2)p(x_r)p(\hat{y}_r|x_r,y_r)$. Randomly and independently generate a codebook for each block $j\in [1:b]$
\begin{itemize}
\item Independently generate $2^{nR_1}$ sequences $x_1^n(m_{1,j})\sim\prod ^n_{i=1}p(x_{1i}) $, where $m_{1,j} \in [1:2^{nR_1}]$.
\item Independently generate $2^{nR_2}$ sequences $x_2^n(m_{2,j})\sim\prod ^n_{i=1}p(x_{2i}) $, where $m_{2,j} \in [1:2^{nR_2}]$.
\item Independently generate $2^{nR_r}$ sequences $x_r^n(k_{j-1})\sim\prod ^n_{i=1}p(x_{ri})$, where $k_{j-1} \in [1:2^{nR_r}]$.
\item For each $k_{j-1}\in [1:2^{nR_r}]$, independently generate $2^{nR_r}$ sequences $\hat{y}^n_r(k_j|k_{j-1})\sim\prod ^n_{i=1}p(\hat{y}_{ri}|x_{ri}(k_{j-1}))$, where $k_{j}\in [1:2^{nR_r}]$.
\end{itemize}

\subsubsection{Encoding}
 User 1 and user 2 respectively transmit $x_1^n(m_{1,j})$ and $x_2^n(m_{2,j})$ in block $j$. The relay, upon receiving $y^n_r(j)$, finds an index $k_j$ such that
 \begin{align*}
(\hat{y}_r^n(k_j|k_{j-1}),y^n_r(j),x^n_r(k_{j-1}))\in T^{(n)}_{\epsilon}.
 \end{align*}
 Assume that such $k_j$ is found, the relay sends $x^n_r(k_j)$ in block $j+1$.\par

\subsubsection{Decoding}: We discuss the decoding at user 1. Assume user 1 has decoded $k_{j-1}$ correctly in block $j$. Then in block $j+1$, user 1 finds a unique pair of $(\hat{m}_{2,j},\hat{k}_j)$ such that
\begin{align}
\!\!\!\!(x_2^n(\hat{m}_{2,j}),x_r^n(\hat{k}_{j-1}),\hat{y}_r^n(\hat{k}_j|\hat{k}_{j-1}),y_1^n(j),x_1^n(m_{1,j}))&\in T^{(n)}_{\epsilon} \nonumber\\
\textrm{and}~~~~~~~~~~~~~~(x_r^n(\hat{k}_j),y_1^n(j+1),x_1^n(m_{1,j+1}))&\in T^{(n)}_{\epsilon}.
\label{decoding_rule_1}
\end{align}

\begin{thm}
\label{rate_two_way_nobin}
The following rate region is achievable for the two-way relay channel by using compress-forward without binning but with joint decoding:
\begin{align}
\label{rate_thm1}
R_1 \leq \min\{&I(X_1;Y_2,\hat{Y}_r|X_2,X_r),\\
&I(X_1,X_r;Y_2|X_2)-I(\hat{Y}_r;Y_r|X_1,X_2,X_r,Y_2)\} \nonumber \\
R_2 \leq \min\{&I(X_2;Y_1,\hat{Y}_r|X_1,X_r),\nonumber \\
&I(X_2,X_r;Y_1|X_1)-I(\hat{Y}_r;Y_r|X_1,X_2,X_r,Y_1)\} \nonumber
\end{align}
subject to
\begin{align}
\label{const_1}
I(\hat{Y}_r;Y_r|X_1,X_2,X_r,Y_1)&\leq I(X_r;Y_1|X_1)\nonumber\\
I(\hat{Y}_r;Y_r|X_1,X_2,X_r,Y_2)&\leq I(X_r;Y_2|X_2)
\end{align}
for some $p(x_1)p(x_2)p(x_r)p(y_1,y_2,y_r|x_1,x_2,x_r)p(\hat{y}_r|x_r,y_r)$.
\end{thm}

\begin{proof}
See Appendix \ref{appendix:rate_two_way_nobin}.
\end{proof}

\subsection{Comparison with the original compress-forward scheme}
In this section, we first present the rate region achieved by the original CF scheme for the two way relay channel as in \cite{rankov2006achievable}. We then show that CF without
 binning but with joint decoding can achieve a larger rate region.\par

We outline the encoding and decoding techniques of the original CF scheme as follows. In each block, each user sends a new message using an independently generated codebook. At the end of each block, the relay finds a description of its received signal from both users. Then it partitions the description index into equal-size bins and sends the codeword for the bin index. Each user applies 3-step successive decoding, in which it first decodes the bin index from the relay, then decodes the compression index within that bin, and at last
decodes the message from the other user.

\begin{thm}
\label{thm_rate_ori}
\em{[Rankov and Wittneben].}
\emph{
The following rate region is achievable for two-way relay channel with the original compress-forward scheme:
\begin{align}
\label{rate_thm2}
R_1& \leq I(X_1;Y_2,\hat{Y}_r|X_2,X_r) \nonumber \\
R_2& \leq I(X_2;Y_1,\hat{Y}_r|X_1,X_r)
\end{align}
subject to
\begin{align}
\label{const_2}
&\max (I(\hat{Y}_r;Y_r|X_1,X_r,Y_1),I(\hat{Y}_r;Y_r|X_2,X_r,Y_2))\nonumber\\
\leq &\min(I(X_r;Y_1|X_1),I(X_r;Y_2|X_2))
\end{align}
for some $p(x_1)p(x_2)p(x_r)p(y_1,y_2,y_r|x_1,x_2,x_r)p(\hat{y}_r|x_r,y_r)$.
}
\end{thm}

We present a short proof of this theorem in Appendix \ref{appendix:thm_rate_ori} to show the difference from CF without binning. The proof follows the same lines as in \cite{rankov2006achievable}, but we also correct an error in the analysis of \cite{rankov2006achievable} as pointed out in Remark \ref{rem:error} in Appendix \ref{appendix:thm_rate_ori}.

\begin{proof}
See Appendix \ref{appendix:thm_rate_ori}.
\end{proof}

\begin{thm}
\label{cor:equ_ori}
In the two-way relay channel, the rate region achieved by compress-forward without binning in Theorem \ref{rate_two_way_nobin} is larger than the rate region achieved by the original compress-forward scheme in Theorem
\ref{thm_rate_ori} when the channel is asymmetric for the two users. The two regions may be equal only if the channel is symmetric, that is the following conditions holds:
\begin{align}
I(X_r;Y_1|X_1)&=I(X_r;Y_2|X_2)\nonumber \\
I(\hat{Y}_r;Y_r|X_1,X_r,Y_1)&=I(\hat{Y}_r;Y_r|X_2,X_r,Y_2).
\label{eq_con}
\end{align}
Furthermore, \eqref{eq_con} is only necessary but may not be sufficient.\footnote{In \cite{zhong2011compress}, we stated condition \eqref{eq_con} as "if and only if", but it should be corrected as "only if". The proof remains the same.}
\end{thm}

\begin{proof}
First, we show that the constraint on the compression rate of Theorem \ref{rate_two_way_nobin} is looser than that of Theorem \ref{thm_rate_ori}. This is true since from (\ref{const_2}), we have
\begin{align}
I(X_r;Y_1|X_1)& \geq I(\hat{Y}_r;Y_r|X_1,X_r,Y_1)\nonumber\\
&=I(\hat{Y}_r;X_2,Y_r|X_1,X_r,Y_1)\nonumber\\
&\geq I(\hat{Y}_r;Y_r|X_1,X_2,X_r,Y_1)
\label{rate_13}
\end{align}
where (\ref{rate_13}) is the right-hand-side of the first term in (\ref{const_1}). Similar for the other term.\par
Next we show that (\ref{rate_thm2}) and (\ref{const_2}) imply (\ref{rate_thm1}). From (\ref{rate_thm2}), we have
\begin{align}
R_2&\leq I(X_2;Y_1,\hat{Y}_r|X_1,X_r)\nonumber \\
&=I(X_2;Y_1|X_1,X_r)+I(\hat{Y}_r;X_2|Y_1,X_1,X_r)\nonumber \\
&=I(X_2,X_r;Y_1|X_1)-I(X_r;Y_1|X_1)+I(\hat{Y}_r;X_2|Y_1,X_1,X_r)\nonumber \\
&\overset{(a)}{\leq} I(X_2,X_r;Y_1|X_1)-I(\hat{Y}_r;Y_r|X_1,X_r,Y_1)
+I(\hat{Y}_r;X_2|Y_1,X_1,X_r)\nonumber \\
&=I(X_2,X_r;Y_1|X_1)-I(\hat{Y}_r;Y_r|X_1,X_2,X_r,Y_1)\nonumber
\end{align}
where $(a)$ follows from the constraint of (\ref{const_2}) in Theorem \ref{thm_rate_ori}. The equality holds when
\begin{align}
I(X_r;Y_1|X_1)&=\min\{I(X_r;Y_1|X_1),I(X_r;Y_2|X_2)\}\nonumber \\
&=I(\hat{Y}_r;Y_r|X_1,X_r,Y_1)\nonumber\\
&=\max (I(\hat{Y}_r;Y_r|X_1,X_r,Y_1),I(\hat{Y}_r;Y_r|X_2,X_r,Y_2)).\nonumber
\end{align}
Similar for $R_1$, the equality holds when
\begin{align}
I(X_r;Y_2|X_2)&=\min\{I(X_r;Y_1|X_1),I(X_r;Y_2|X_2)\}\nonumber \\
&=I(\hat{Y}_r;Y_r|X_2,X_r,Y_2)\nonumber\\
&=\max (I(\hat{Y}_r;Y_r|X_1,X_r,Y_1),I(\hat{Y}_r;Y_r|X_2,X_r,Y_2)).\nonumber
\end{align}
The above analysis shows that at the boundary of the compression rate constraint \eqref{const_2}, the rate region of CF without binning \eqref{rate_thm1} is equivalent to that of the original CF \eqref{rate_thm2}. However, since constraint \eqref{const_1} is loser than \eqref{const_2}, the rate region in Theorem \ref{rate_two_way_nobin} is larger than that in Theorem \ref{thm_rate_ori}.
Only if condition (\ref{eq_con}) holds,
the original CF scheme may achieve the same rate region as CF without binning; otherwise, its rate region is strictly smaller.
\end{proof}
\begin{rem}
For the two-way relay channel, binning and successive decoding constrains the compression rate to the weaker of the channels from relay to two users. But without binning, this constraint is relaxed.
Thus CF without binning achieves a larger rate region than the original CF scheme when then channel is asymmetric for the two users.
\end{rem}

\subsection{Comparison with Noisy Network Coding}
In this section, we compare the rate region achieved by CF without binning with that by noisy network coding \cite{sung2011noisy} for the two way relay channel.
The main differences between these two schemes are as follows. In CF without binning, different messages are sent over different blocks, but in noisy network coding, the same message is sent in multiple blocks using independent codebooks. Furthermore, in noisy network coding, each user performs simultaneous joint decoding of the message based on signals received from all blocks without uniquely decoding the compression indices (i.e. relaxed joint decoding). But in CF without binning, each user jointly decodes both the message and compression index precisely based on signals received from the current and previous blocks.
\begin{thm}
\label{thm_noisy}
\em{[Lim, Kim, El Gamal and Chung].}
\emph{
The following rate region is achievable for the two-way relay channel with noisy network coding:
\begin{align}
\label{rate_noisy}
R_1 \leq \min\{&I(X_1;Y_2,\hat{Y}_r|X_2,X_r), \\
&I(X_1,X_r;Y_2|X_2)-I(\hat{Y}_r;Y_r|X_1,X_2,X_r,Y_2)\} \nonumber \\
R_2 \leq \min\{&I(X_2;Y_1,\hat{Y}_r|X_1,X_r),\nonumber \\
&I(X_2,X_r;Y_1|X_1)-I(\hat{Y}_r;Y_r|X_1,X_2,X_r,Y_1)\}\nonumber
\end{align}
for some $p(x_1)p(x_2)p(x_r)p(y_1,y_2,y_r|x_1,x_2,x_r)p(\hat{y}_r|x_r,y_r)$.
}
\end{thm}
Comparing Theorem \ref{rate_two_way_nobin} with Theorem \ref{thm_noisy}, we find that the rate constraints (\ref{rate_thm1}) for $R_1$ and $R_2$ in CF without binning are the same as those in noisy network coding (\ref{rate_noisy}). However, CF without binning has an extra constraint on the compression rate (\ref{const_1}). Therefore, in general, CF without binning achieves a smaller rate region than noisy network coding. Sections \ref{sec:Gaussian_TWRC} and \ref{sec:Fading_TWRC} show that for the Gaussian and fading two-way relay channels, these two schemes achieve the same rate region or sum rate for certain channel configurations.\par
\begin{rem}
In noisy network coding, the combination of message repetition and relaxed joint decoding is necessary in addition to compress-forward without binning to achieve a better rate region than that in Theorem \ref{rate_two_way_nobin} for the TWRC.
\end{rem}

\begin{figure}[!t]
\centering
  \includegraphics[scale=0.6]{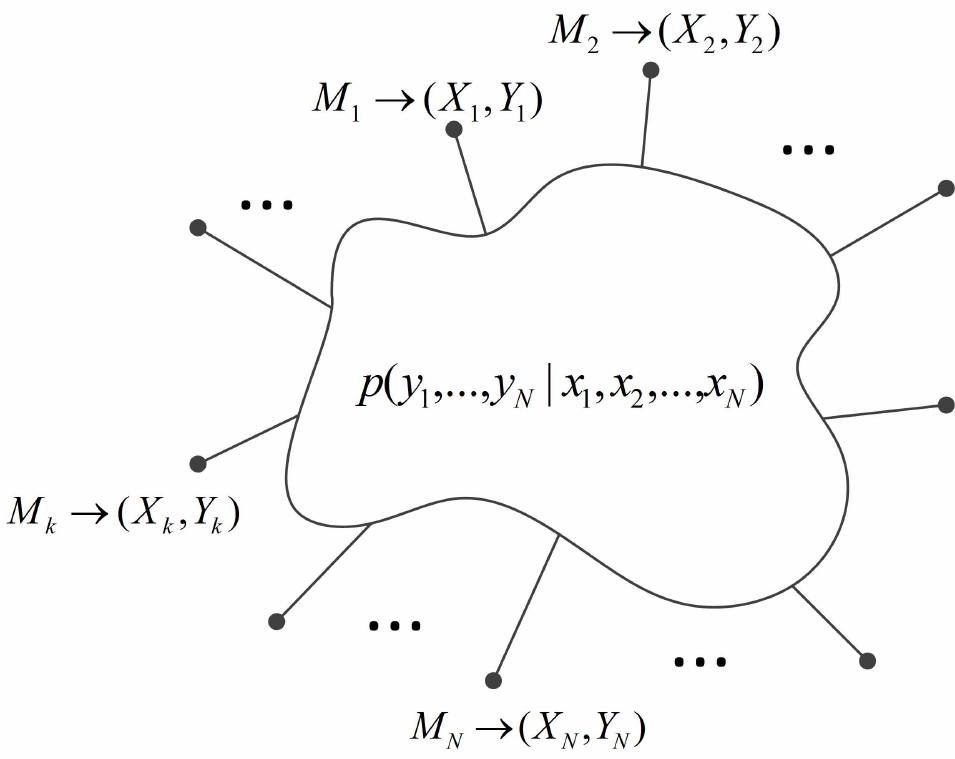}\\
  \caption{Relay network model.}\label{relay_network_model}
\end{figure}
\section{Implication for Relay Networks}\label{sec:relay_networks}
From the discussion for the one-way and two-way relay channels, we can acquire some implications for the general relay networks. A $N$-node discrete memoryless relay network $(\prod ^N_{k=1}\mathcal{X}_k, p(y_N|x_N),\times^N_{k=1}\mathcal{Y}_k)$ is depicted in Figure \ref{relay_network_model}. It consists of $N$ sender-receiver alphabet pairs $(\mathcal{X}_k,\mathcal{Y}_k)$, $k\in [1:N]$ and a collection of conditional pmfs $p(y_1,\cdots y_N|x_1,\cdots x_N)$. Each node $k\in[1:N]$ wishes to send a message $M_k$ to a set of destination nodes,
while also acting as a relay for messages from other nodes. \par
Noisy network coding is proposed for the general relay network in \cite{sung2011noisy}. Three new ideas are used in noisy network coding. One is no Wyner-Ziv binning in relay operation. Another is simultaneous joint decoding of the message over all blocks without uniquely decoding the compression indices. Last is message repetition, in which the same message is sent in multiple blocks using independent codebooks. Next, we discuss the effect of each of these ideas separately.

\subsection{Implication of no Wyner-Ziv binning}
Generalizing from the two-way relay channel, we can conclude that CF without binning achieves a larger rate region than the original CF scheme for networks with multiple destinations. With binning and successive decoding, the compression rate is constrained by the weakest link from a relay to a destination, as in (\ref{rate_twoway_ori_3}). But with joint decoding of the message and compression indices, this constraint is more relaxed since each destination can also use the signals received from other links, including the direct links, to decode the compression indices and provides relays more freedom to choose the compression rates. This explains why for the two-way relay channel, the constraint on compression rate in CF without binning (\ref{const_1}) is looser than that in the original CF scheme (\ref{const_2}). Therefore, with joint decoding, Wyner-Ziv binning is not necessary. No binning also simplifies relay operation.

\begin{rem}
Joint decoding of both the message and compression index is crucial for CF without binning. Without joint decoding, it achieves strictly smaller rate than the original CF with binning for any network.
\end{rem}

\subsection{Implication of joint decoding without explicitly decoding compression indices}
A difference between CF without binning and noisy network coding comes from the decoding of the compression indices. For both schemes, the compression rate at a relay is lower bounded by the covering lemma. But in noisy network coding, each destination does not decode the compression indices explicitly, hence there are no additional constraints on the compression rate. However, in CF without binning, each destination decodes the compression indices precisely; this decoding places extra upper bounds on the compression rate, leading to the constraints on compression rate as in (\ref{const_1}).\par
The above analysis prompts the question: what if in CF without binning, we also do not decode the compression index precisely, can we achieve the same rate region as noisy network coding? The following analysis shows that the answer is negative. Take the two-way relay channel as an example. In the next few sections, we apply several joint decoding rules to enlarge the rate region.

\subsubsection{Relaxed joint decoding of a single message without decoding the compression indices uniquely}
Using the same codebook generation and encoding as in CF without binning in Section \ref{sec:TWRC_CF_withoutbinning}, but we change the decoding rule in \eqref{decoding_rule_1} to as follows. In block $j+1$, user 1 finds a unique $\hat{m}_{2,j}$ such that
\begin{align}
(x_2^n(\hat{m}_{2,j}),x_r^n(\hat{k}_{j-1}),\hat{y}_r^n(\hat{k}_j|\hat{k}_{j-1}),y_1^n(j),x_1^n(m_{1,j}))&\in T^{(n)}_{\epsilon} \nonumber\\
\textrm{and}~~~~~~~~~~~~~~(x_r^n(\hat{k}_j),y_1^n(j+1),x_1^n(m_{1,j+1}))&\in T^{(n)}_{\epsilon}
\label{decoding_rule_no_comp}
\end{align}
for some pair of indices $(\hat{k}_{j-1},\hat{k}_{j})$.
With this decoding rule, the error event $\mathcal{E}_{3j}$ which corresponds to wrong compression index only (see Appendix \ref{appendix:rate_two_way_nobin}) no longer applies, but other new error events appear in which $\hat{k}_{j-1} \neq 1$ (see the detailed error analysis in Appendix \ref{appendix:cor:no_index_decoding}). With decoding rule \eqref{decoding_rule_no_comp}, we have following Corollary:
\begin{cor}
\label{cor:no_index_decoding}
For the two-way relay channel, the following rate region is achievable by CF without binning using joint decoding but without decoding the compression index precisely:
\begin{subequations}
\label{rate:no_repeat}
\begin{align}
R_1 &\leq I(X_1;Y_2,\hat{Y}_r|X_2,X_r) \\
R_1 &\leq I(X_1,X_r;Y_2|X_2)-I(\hat{Y}_r;Y_r|X_1,X_2,X_r,Y_2)  \\
\label{two_more_1}
R_1 &\leq I(X_1,X_r;Y_2|X_2)-I(\hat{Y}_r;Y_r|X_1,X_2,X_r,Y_2)+I(X_r;Y_2|X_2)-I(\hat{Y};Y_2|X_r) \\
R_2 &\leq I(X_2;Y_1,\hat{Y}_r|X_1,X_r) \\
R_2 &\leq I(X_2,X_r;Y_1|X_1)-I(\hat{Y}_r;Y_r|X_1,X_2,X_r,Y_1) \\
\label{two_more_2}
R_2 &\leq I(X_2,X_r;Y_1|X_1)-I(\hat{Y}_r;Y_r|X_1,X_2,X_r,Y_1)+I(X_r;Y_1|X_1)-I(\hat{Y};Y_1|X_r)
\end{align}
\end{subequations}
for some $p(x_1)p(x_2)p(x_r)p(y_1,y_2,y_r|x_1,x_2,x_r)p(\hat{y}_r|x_r,y_r)$.\par
\end{cor}
\begin{proof}
See Appendix \ref{appendix:cor:no_index_decoding}.
\end{proof}
In above rate region, although the new decoding rule \eqref{decoding_rule_no_comp} removes the constraint on compression
 rate, it brings two new rate constraints \eqref{two_more_1} \eqref{two_more_2}.
 Hence the rate region in \eqref{rate:no_repeat} is still smaller than that of noisy network coding in (\ref{rate_noisy}).
  These two extra rate constraints come from the block boundary condition when performing joint decoding over two blocks.
  Specifically, it corresponds to the error event $\mathcal{E}_{7j}$ in Appendix \ref{appendix:cor:no_index_decoding},
  in which all the compression indices and message are wrong. The boundary condition results from the second decoding rule in \eqref{decoding_rule_no_comp} in which not all input and output signals are involved.\par

\subsubsection{Simultaneous joint decoding of all messages without decoding compression indices uniquely}
  To reduce the boundary effect, we try simultaneous decoding of all messages over all blocks, again without decoding the compression indices uniquely. We use the same codebook generation and encoding as in CF without binning in Section \ref{sec:TWRC_CF_withoutbinning}, but use a different decoding rule at the destination.
  The destination now jointly decodes all messages based on the signals received in all blocks, without decoding the compression indices explicitly.
  Specifically, user 1 now finds a unique tuple $(m_{2,1},\dots m_{2,j},\dots m_{2,b-1})$ at the end of block $b$ such that
  \begin{subequations}
  \begin{align}
  \label{simul_decoding_10}
&(x_2^n({m}_{2,1}),x_r^n(1),\hat{y}_r^n({k}_1|1),y_1^n(1),x_1^n(m_{1,1}))\in T^{(n)}_{\epsilon}\\
&(x_2^n({m}_{2,2}),x_r^n(k_1),\hat{y}_r^n({k}_2|k_1),y_1^n(2),x_1^n(m_{1,2}))\in T^{(n)}_{\epsilon}\nonumber \\
\vdots\nonumber\\
&(x_2^n({m}_{2,j}),x_r^n(k_{j-1}),\hat{y}_r^n({k}_j|k_{j-1}),y_1^n(j),x_1^n(m_{1,j}))\in T^{(n)}_{\epsilon}\nonumber \\
\vdots\nonumber\\
\label{simul_decoding_40}
&(x_2^n({m}_{2,b-1}),x_r^n(k_{b-2}),\hat{y}_r^n({k}_{b-1}|k_{b-2}),y_1^n(b-1),x_1^n(m_{1,b-1}))\in T^{(n)}_{\epsilon}\\
\label{simul_decoding_50}
&(x_2^n(1),x_r^n(k_{b-1}),\hat{y}_r^n({k}_{b}|k_{b-1}),y_1^n(b),x_1^n(1))\in T^{(n)}_{\epsilon}
\end{align}
\end{subequations}
  for some indices $(k_1,\dots k_j,\dots k_{b-1},k_{b})$. \par
  To compare its achievable rate region with noisy network coding, we consider the error event in which $m_{2,1}$ and $k_{b}$ are wrong while all other messages and indices are right. This error event involves the decoded joint distributions from \eqref{simul_decoding_10} and \eqref{simul_decoding_50} as follows.
\begin{align*}
p(x_1)p(x_2)p(x_r)p(y_1,\hat{y}_r|x_1,x_r)\\
p(x_1)p(x_2)p(x_r)p(\hat{y}_r|x_r)p(y_1|x_1,x_2,x_r)
\end{align*}
By the packing lemma, the probability of the above error event goes to zero as $n \rightarrow\infty$ if
\begin{align}
R_2+R_r \leq I(X_2;Y_1,\hat{Y}_r|X_1,X_r)+I(\hat{Y}_r;X_1,X_2,Y_1|X_r)
\end{align}
 Since $R_r >I(\hat{Y}_r;Y_r|X_r)$ by the covering lemma \cite{gamal2010lecture}, we obtain
\begin{align}
R_2 &\leq I(X_2;Y_1,\hat{Y}_r|X_1,X_r)+I(\hat{Y}_r;X_1,X_2,Y_1|X_r)
-I(\hat{Y}_r;Y_r|X_r)\nonumber\\
&= I(X_2;Y_1,\hat{Y}_r|X_1,X_r)-I(\hat{Y}_r;Y_r|X_1,X_2,X_r,Y_1)
\label{tight_const_0}
\end{align}
  which is tighter than the constraint \eqref{rate_noisy} in noisy network coding. Again we see the boundary effect at the last block when $k_b$ is wrong.

\subsubsection{Simultaneous joint decoding of all messages but ignoring the last compression index}
  Since in the last block, each source sends a known message, the last compression index $k_b$ brings no new information. Hence we may choose to omit it in the decoding rule to see if the rate region can be improved.
  Specifically, user 1 now finds a unique tuple $(m_{2,1},\dots m_{2,j},\dots m_{2,b-1})$ at the end of block $b$ such that
\begin{subequations}
\begin{align}
\label{simul_decoding_1}
&(x_2^n({m}_{2,1}),x_r^n(1),\hat{y}_r^n({k}_1|1),y_1^n(1),x_1^n(m_{1,1}))\in T^{(n)}_{\epsilon}\\
&(x_2^n({m}_{2,2}),x_r^n(k_1),\hat{y}_r^n({k}_2|k_1),y_1^n(2),x_1^n(m_{1,2}))\in T^{(n)}_{\epsilon}\nonumber \\
\vdots\nonumber\\
&(x_2^n({m}_{2,j}),x_r^n(k_{j-1}),\hat{y}_r^n({k}_j|k_{j-1}),y_1^n(j),x_1^n(m_{1,j}))\in T^{(n)}_{\epsilon}\nonumber \\
\vdots\nonumber\\
\label{simul_decoding_4}
&(x_2^n({m}_{2,b-1}),x_r^n(k_{b-2}),\hat{y}_r^n({k}_{b-1}|k_{b-2}),y_1^n(b-1),x_1^n(m_{1,b-1}))\in T^{(n)}_{\epsilon}\\
\label{simul_decoding_5}
&(x_2^n(1),x_r^n(k_{b-1}),y_1^n(b),x_1^n(1))\in T^{(n)}_{\epsilon}
\end{align}
\end{subequations}
for some indices $(k_1,\dots k_j,\dots k_{b-1})$.
 Note that \eqref{simul_decoding_5} is the only step that is different from \eqref{simul_decoding_50}.\par
 To compare its achievable rate region with noisy network coding, we consider the error event in which $m_{2,1}$ and $k_{b-1}$ are wrong while all other messages and indices are right. This error event involves the decoded joint distributions from \eqref{simul_decoding_1}, \eqref{simul_decoding_4} and \eqref{simul_decoding_5} as follows.
\begin{align*}
p(x_1)p(x_2)p(x_r)p(y_1,\hat{y}_r|x_1,x_r)\\
p(x_1)p(x_2)p(x_r)p(\hat{y}_r|x_r)p(y_1|x_1,x_2,x_r)\\
p(x_1)p(x_2)p(x_r)p(y_1|x_1,x_2)
\end{align*}
By the packing lemma, the probability of the above error event goes to zero as $n \rightarrow\infty$ if
\begin{align}
R_2+R_r \leq I(X_2;Y_1,\hat{Y}_r|X_1,X_r)
 +I(X_r;Y_1|X_1,X_2)+I(\hat{Y}_r;X_1,X_2,Y_1|X_r).\nonumber
\end{align}
Since $R_r >I(\hat{Y}_r;Y_r|X_r)$, we obtain
\begin{align}
R_2 &\leq I(X_2;Y_1,\hat{Y}_r|X_1,X_r)+I(X_r;Y_1|X_1,X_2)+I(\hat{Y}_r;X_1,X_2,Y_1|X_r)-I(\hat{Y}_r;Y_r|X_r)\nonumber\\
\label{tight_const}
&= I(X_2;Y_1,\hat{Y}_r|X_1,X_r)+I(X_r;Y_1|X_1,X_2)-I(\hat{Y}_r;Y_r|X_1,X_2,X_r,Y_1)
\end{align}
Although rate constraint \eqref{tight_const} is loser than \eqref{tight_const_0}, it is still
tighter than the rate constraint of noisy network coding in \eqref{rate_noisy} (we can easily check this for the Gaussian TWRC). Therefore, simultaneous decoding of all messages over all blocks still achieves a smaller rate region than noisy network coding.

\begin{rem}
Joint decoding without explicitly decoding the compression indices removes the constraints on compression rates but still achieves a smaller rate region than noisy network coding for the general multi-source network.
\end{rem}

\begin{rem}
Simultaneous joint decoding of all messages over all blocks still does not overcome the block boundary limitation, caused by the last block in which no new messages but only the compression indices are sent.
\end{rem}

\begin{rem}
In \cite{Wu2010on}, a technique of appending $M$ blocks for transmission of the compression indices only is proposed to reduce the block boundary effect, but it also reduces the achievable rate by a non-vanishing amount.
\end{rem}

The above analysis leads us to discuss the effect of message repetition, in which the same message is transmitted over multiple blocks.

\begin{table*}[!t]
\centering
\begin{spacing}{1.5}
\begin{tabular}{|c|c|c|c|c|c|}
\hline
Block&\ldots&$2j-1$&$2j$&$2j+1$&\ldots  \\
\hline
$X_1$ & \ldots & $x^n_{1,2j-1}(m_{1,j})$  &$x^n_{1,2j}(m_{1,j})$  & $x^n_{1,2j+1}(m_{1,j+1})$  &\ldots \\
\hline
$X_2$ & \ldots & $x^n_{2,2j-1}(m_{2,j})$  &$x^n_{2,2j}(m_{2,j})$  & $x^n_{2,2j+1}(m_{2,j+1})$  &\ldots \\
\hline
$Y_r$ & \ldots & $\hat{y}_r(k_{2j-1}|k_{2j-2})$ &$\hat{y}_r(k_{2j}|k_{2j-1})$ & $\hat{y}_r(k_{2j+1}|k_{2j})$ &\ldots \\
\hline
$X_r$ & \ldots& $x^n_r({k_{2j-2}})$ & $x^n_r({k_{2j-1}})$ &$x^n_r({k_{2j}})$ &\ldots \\
\hline
$Y_1$ & \ldots & $\hat{m}_{2,j-1}$ for some $(k_{2j-4},k_{2j-3},k_{2j-2})$&--&$\hat{m}_{2,j}$ for some $(k_{2j-2},k_{2j-1},k_{2j})$&\ldots \\
\hline
$Y_2$ & \ldots & $\hat{m}_{1,j-1}$ for some $(k_{2j-4},k_{2j-3},k_{2j-2})$&--&$\hat{m}_{1,j}$ for some $(k_{2j-2},k_{2j-1},k_{2j})$ &\ldots \\
\hline
\end{tabular}
\end{spacing}
\caption{\textnormal{Encoding and decoding of CF without binning but with twice message repetition for the two-way relay channel.}}
\label{tab:CF_two-way_rep}
\end{table*}

\subsection{Implication of message repetition}
Message repetition is performed in noisy network coding by sending the same message in multiple, consecutive blocks using independent codebooks. To understand the effect of message repetition, we use the same decoding rule as in the previous section (decoding the message without explicitly decoding the compression indices) and repeat the message twice. Now each user transmits the same message in every two consecutive blocks. We compare the achievable rate region of this scheme with those of no message repetition and of noisy network coding, in which each message is repeated $b$ times and $b$ approaches to infinity.\par

Again take the two-way relay channel as an illustrative example.
Let each user transmit the same message in every two blocks using independent codebooks. Decoding is performed at the end of every two blocks based on signals received from the current and previous two blocks. Specifically, the codebook generation, encoding and decoding are as follows (also see Table \ref{tab:CF_two-way_rep}). We use a block coding scheme in which each user sends $b$ messages over $2b+1$ blocks of $n$ symbols each.\par
\subsubsection{Codebook generation}
At block $l\in\{2j-1,2j\}$:
\begin{itemize}
\item Independently generate $2^{n2R_1}$ sequences $x_{1,l}^n(m_{1,j})\sim\prod ^n_{i=1}p(x_{1i})$, where $m_{1,j} \in [1:2^{n2R_1}]$.
\item Independently generate $2^{n2R_2}$ sequences $x_{2,l}^n(m_{2,j})\sim\prod ^n_{i=1}p(x_{2i})$, where $m_{2,j} \in [1:2^{n2R_2}]$.
\item Independently generate $2^{nR_r}$ sequences $x_r^n(k_{l})\sim\prod ^n_{i=1}p(x_{ri})$, where $k_{l} \in [1:2^{nR_r}]$.
\item For each $k_{l}\in [1:2^{nR_r}]$, independently generate $2^{nR_r}$ sequences $\hat{y}^n_r(k_{l+1}|k_{l})\sim\prod ^n_{i=1}p(\hat{y}_{ri}|x_{ri}(k_{l}))$, where $k_{l+1}\in [1:2^{nR_r}]$.
\end{itemize}

\subsubsection{Encoding}
In blocks $2j-1$ and $2j$, user 1 transmits $x_{1,2j-1}^n(m_{1,j})$ and $x_{1,2j}^n(m_{1,j})$ respectively. User 2 transmits $x_{2,2j-1}^n(m_{2,j})$ and $x_{2,2j}^n(m_{2,j})$.\par
In block $j$, the relay, upon receiving $y^n_r(j)$, finds an index $k_j$ such that
\begin{align*}
(\hat{y}_r^n(k_j|k_{j-1}),y^n_r(j),x^n_r(k_{j-1}))\in T^{(n)}_{\epsilon}.
 \end{align*}
 Assume that such $k_j$ is found, the relay sends $x^n_r(k_j)$ in block $j+1$.\par

\subsubsection{Decoding} We discuss the decoding at user 1. User 1 decodes $m_{2,j}$ at the end of block $2j+1$. Specifically, it finds a unique $\hat{m}_{2,j}$ such that
\begin{align}
\label{decode_rule_repeat}
(x_{2,2j-1}^n(\hat{m}_{2,j}),x_r^n(\hat{k}_{2j-2}),\hat{y}_r^n(\hat{k}_{2j-1}|\hat{k}_{2j-2}),
y_1^n(2j-1),x_{1,2j-1}^n(m_{1,j}))&\in T^{(n)}_{\epsilon} \nonumber\\
(x_{2,2j}^n(\hat{m}_{2,j}),x_r^n(\hat{k}_{2j-1}),\hat{y}_r^n(\hat{k}_{2j}|\hat{k}_{2j-1}),
y_1^n(2j),x_{1,2j}^n(m_{1,j}))&\in T^{(n)}_{\epsilon} \nonumber\\
\!\!\!\!\textrm{and}~~~~~~~(x_r^n(\hat{k}_{2j}),y_1^n(2j+1),x_{1,2j+1}^n(m_{1,j+1}))&\in T^{(n)}_{\epsilon}
\end{align}
for some $(\hat{k}_{2j-2},\hat{k}_{2j-1},\hat{k}_{2j})$.

\begin{cor}
\label{cor:repeat_once_decoding}
For the two-way relay channel, the following rate region is achievable by CF without binning and without explicitly decoding the compression indices when repeating each message twice:
\begin{subequations}
\label{rate:repeat_once}
\begin{align}
R_1 &\leq I(X_1;Y_2,\hat{Y}_r|X_2,X_r) \\
R_1 &\leq I(X_1,X_r;Y_2|X_2)-I(\hat{Y}_r;Y_r|X_1,X_2,X_r,Y_2)  \\
\label{repeat_two_more_1}
R_1 &\leq I(X_1,X_r;Y_2|X_2)-I(\hat{Y}_r;Y_r|X_1,X_2,X_r,Y_2)
+\frac{1}{2}[I(X_r;Y_2|X_2)-I(\hat{Y};Y_2|X_r)] \\
R_2 &\leq I(X_2;Y_1,\hat{Y}_r|X_1,X_r) \\
R_2 &\leq I(X_2,X_r;Y_1|X_1)-I(\hat{Y}_r;Y_r|X_1,X_2,X_r,Y_1) \\
\label{repeat_two_more_2}
R_2 &\leq I(X_2,X_r;Y_1|X_1)-I(\hat{Y}_r;Y_r|X_1,X_2,X_r,Y_1)
+\frac{1}{2}[I(X_r;Y_1|X_1)-I(\hat{Y};Y_1|X_r)]
\end{align}
\end{subequations}
for some $p(x_1)p(x_2)p(x_r)p(y_1,y_2,y_r|x_1,x_2,x_r)p(\hat{y}_r|x_r,y_r)$.
\end{cor}
\begin{proof}
See Appendix \ref{appendix:cor:repeat_once_decoding}.
\end{proof}

Comparing the above rate region with that of CF without message repetition in \eqref{rate:no_repeat} and noisy network coding in (\ref{rate_noisy}), we find that the extra rate constraints \eqref{two_more_1} and \eqref{two_more_2} are relaxed by repeating the message twice. The additional term is divided by 2, which comes from the error event $\mathcal{E}^r_{10j}$ in Appendix \ref{appendix:cor:repeat_once_decoding}, in which the message and all compression indices are wrong. This event also corresponds to a boundary event, but since decoding rule \eqref{decode_rule_repeat} spans more blocks, the boundary effect is lessen.\par
Thus if we repeat the messages $b$ times, this additional term will be divided by $b$. Taking $b$ to infinity as in noisy network coding completely eliminates the additional terms in \eqref{repeat_two_more_1} and \eqref{repeat_two_more_2}. Hence the rate region is increasing with the number of times for message repetition.
To achieve the largest rate region, the message repetition times need to be infinity. Therefore, noisy network coding has $b$ blocks decoding delay and $b$ is required to approach infinity.

\begin{rem}
Message repetition brings more correlation between different blocks. It helps lessen the boundary effect when increasing the repetition times.
\end{rem}

\begin{rem}
Message repetition is necessary for achieving the noisy network coding rate region in multi-source networks. Furthermore, the repetition times need to approach infinity.
\end{rem}

\begin{rem}
This result is different from the single-source single-destination result in \cite{Wu2010on} which shows backward decoding without message repetition achieves the same rate as noisy network coding, albeit at an expense of extending the relay forwarding times infinitely without actually sending a new message. Thus message repetition appears to be essential in a multi-source network.
\end{rem}


\section{Gaussian Two-way Relay Channel}\label{sec:GTWRC}
\label{sec:Gaussian_TWRC}

\begin{figure}[!t]
    \begin{center}
    \includegraphics[width=0.3\textwidth]{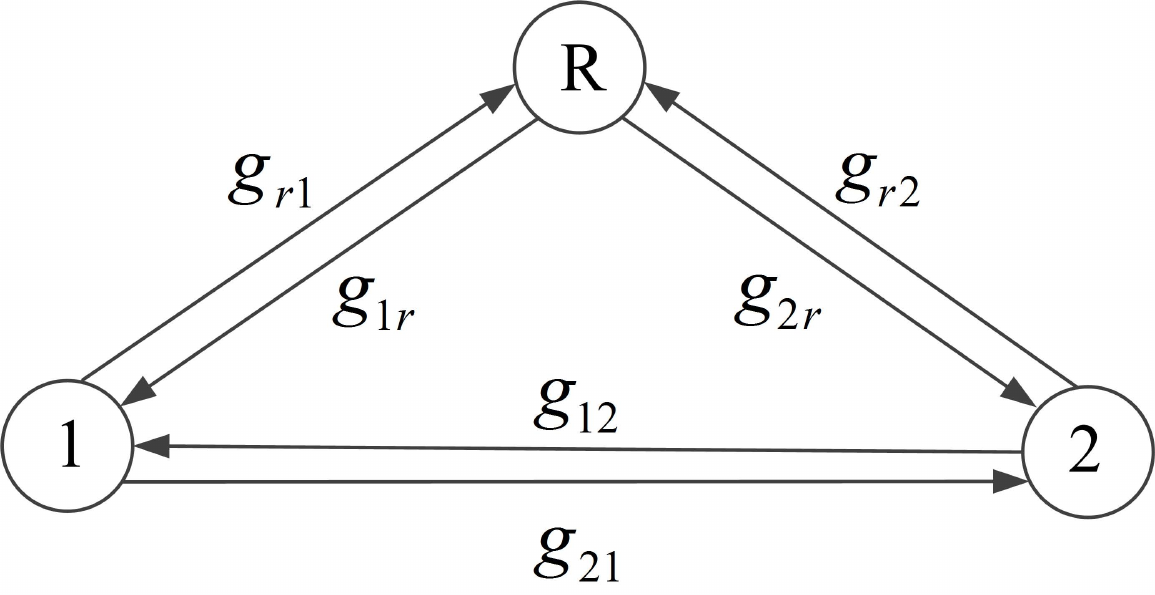}
    \caption{Gaussian two-way relay channel model.} \label{fig:GTWRC}
    \end{center}
\end{figure}

In this section, we focus on the Gaussian two-way relay channel. We compare the rate regions of CF without binning, the original CF and noisy network coding. We also derive the analytical condition for when CF without binning achieves the same rate region or sum rate as noisy network coding.\par
As in Figure \ref{fig:GTWRC}, the Gaussian two-way relay channel can be modeled as:
\begin{align}
  Y_1&=g_{12}X_2+g_{1r}X_r+Z_1 \nonumber \\
  Y_2&=g_{21}X_1+g_{2r}X_r+Z_2 \nonumber \\
  Y_r&=g_{r1}X_1+g_{r2}X_2+Z_r
  \label{GTWRC}
\end{align}
where $Z_1, Z_2, Z_r\sim\mathcal{N}(0,1)$ are independent Gaussian noises. The average
input power constraints for user 1, user 2 and the relay are all
$P$. $g_{12}, g_{1r}, g_{21}, g_{2r}, g_{r1}, g_{r2}$ are corresponding channel gains.

\subsection{Achievable rate region by CF without binning}
In the Gaussian two-way relay channel model, assume
\begin{align}
\label{Gaussian_signaling}
&X_1 \sim \mathcal{N}(0,P), X_2 \sim \mathcal{N}(0,P),\nonumber\\
&X_r \sim \mathcal{N}(0,P), \hat{Z} \sim \mathcal{N}(0,\sigma^2),
\end{align}
where $X_1, X_2, X_r$ and $\hat{Z}$ are independent, and $\hat{Y}_r=Y_r+\hat{Z}$. Denote
\begin{align}
\label{four_rate}
R_{11}(\sigma^2)&=C\left(g_{21}^2P+\frac{g_{r1}^2P}{1+\sigma^2}\right)\nonumber\\
R_{12}(\sigma^2)&=C(g_{21}^2P+g_{2r}^2P)-C(1/\sigma^2)\nonumber\\
R_{21}(\sigma^2)&=C\left(g_{12}^2P+\frac{g_{r2}^2P}{1+\sigma^2}\right)\nonumber\\
R_{22}(\sigma^2)&=C(g_{12}^2P+g_{1r}^2P)-C(1/\sigma^2),
\end{align}
where $C(x)=\frac{1}{2}\log(1+x).$
Then we have following rate region using CF without binning.
\begin{thm}
\label{cor:G_nobin}
The following rate region is achievable for the Gaussian two-way relay channel using compress-forward without binning:
\begin{align}
\label{Gaussian_compress}
R_1&\leq\min\{R_{11}(\sigma^2),R_{12}(\sigma^2)\}\nonumber\\
R_2&\leq\min\{R_{21}(\sigma^2),R_{22}(\sigma^2)\}
\end{align}
for some $\sigma^2 \geq \max \{\sigma_{c1}^2,\sigma_{c2}^2\}$, where
\begin{align}
\label{c1c2}
\sigma_{c1}^2&=(1+g_{21}^2P)/(g_{2r}^2P)\nonumber\\
\sigma_{c2}^2&=(1+g_{12}^2P)/(g_{1r}^2P),
\end{align}
and $R_{11}(\sigma^2), R_{12}(\sigma^2), R_{21}(\sigma^2), R_{22}(\sigma^2)$ are as defined in (\ref{four_rate}).
\end{thm}
\begin{proof}
Applying the rate region in Theorem \ref{rate_two_way_nobin} with the signaling in \eqref{Gaussian_signaling}, we obtain \eqref{Gaussian_compress}.
\end{proof}

\subsection{Rate region comparison with the original CF scheme}
We now compare the rate regions of CF without binning in Theorem \ref{cor:G_nobin} and the original CF in \cite{rankov2006achievable} for the Gaussian two-way relay channel. We first present the rate region achieved by the original CF scheme.
\begin{cor}
\label{oricf_GTWRC}
\em{[Rankov and Wittneben].}
\emph{
The following rate region is achievable for the Gaussian two-way relay channel using the original compress-forward scheme:
\begin{align}
\label{Gaussian_original}
R_1&\leq R_{11}(\sigma^2) \nonumber\\
R_2&\leq R_{21}(\sigma^2)
\end{align}
for some $\sigma^2\geq\sigma_{r}^2$, where
\begin{align}
\sigma_{r}^2=\max \left\{ \frac{1+g_{21}^2P+g_{r1}^2P}{\min\{g_{2r}^2,g_{1r}^2\}P},\frac{1+g_{12}^2P+g_{r2}^2P}{\min\{g_{2r}^2,g_{1r}^2\}P}\right\}
\end{align}
and $R_{11}(\sigma^2), R_{21}(\sigma^2)$ are as defined in (\ref{four_rate}).
}
\end{cor}
The following result shows the condition for which the original CF achieves the same rate region as CF without binning.
\begin{thm}
\label{thm:ori_equ_nobin}
The original compress-forward scheme achieves the same rate region as compress-forward without binning for the Gaussian TWRC if and only if
\begin{align}
\label{ori_equ_nobin}
g_{1r}&=g_{2r}\nonumber\\
g_{21}^2+g_{r1}^2&=g_{12}^2+g_{r2}^2.
\end{align}
Otherwise the rate region by the original compress-forward scheme is smaller.
\end{thm}

\begin{rem}
Condition \eqref{ori_equ_nobin} for the Gaussian TWRC is both sufficient and necessary, and hence is stricter than the result for the DMC case in Theorem \ref{cor:equ_ori}, which is only necessary.
\end{rem}

\begin{proof}
Note that both $R_{11}(\sigma^2),R_{21}(\sigma^2)$ in \eqref{four_rate} are non-increasing and $R_{12}(\sigma^2),R_{22}(\sigma^2)$ are non-decreasing. Let $\sigma_{e1}^2$ be the intersection between $R_{11}(\sigma^2)$ and $R_{21}(\sigma^2)$ (similar for $\sigma_{e2}^2$) as:
\begin{align}
R_{11}(\sigma_{e1}^2)&=R_{12}(\sigma_{e1}^2)\nonumber\\
R_{21}(\sigma_{e2}^2)&=R_{22}(\sigma_{e2}^2).\nonumber
\end{align}
Then we can easily show that
\begin{align}
\label{e1e2}
\sigma_{e1}^2&=(1+g_{21}^2P+g_{r1}^2P)/(g_{2r}^2P)\nonumber\\
\sigma_{e2}^2&=(1+g_{12}^2P+g_{r2}^2P)/(g_{1r}^2P).
\end{align}
Therefore, $\min\{R_{11}(\sigma^2),R_{12}(\sigma^2)\}$ is maximized when $\sigma^2=\sigma_{e1}^2$, while $\min\{R_{21}(\sigma^2),R_{22}(\sigma^2)\}$ is maximized when $\sigma^2=\sigma_{e2}^2$. Noting that $\sigma_{r}^2\geq \sigma_{e1}^2, \sigma_{r}^2\geq \sigma_{e2}^2$ and $\sigma_{e1}^2\geq\sigma_{c1}^2, \sigma_{e2}^2\geq\sigma_{c2}^2$, we conclude that the rate region in Theorem \ref{cor:G_nobin} is the same as that in Corollary \ref{oricf_GTWRC} if and only if $\sigma_{r}^2=\sigma_{e1}^2=\sigma_{e2}^2$. From this equality, we obtain the result in Theorem \ref{thm:ori_equ_nobin}.
\end{proof}

\subsection{Rate region comparison with Noisy Network Coding}
We next compare the rate regions of CF without binning and noisy network coding \cite{sung2011noisy} for the Gaussian two-way relay channel. We first present the rate region achieved by noisy network coding.
\begin{cor}
\em{[Lim, Kim, El Gamal and Chung].}
\emph{
The following rate region is achievable for the Gaussian two-way relay channel with noisy network coding:
\begin{align}
\label{Gaussian_noisy}
R_1&\leq\min\{R_{11}(\sigma^2),R_{12}(\sigma^2)\}\nonumber\\
R_2&\leq\min\{R_{21}(\sigma^2),R_{22}(\sigma^2)\}
\end{align}
for some $\sigma^2>0$, where $R_{11}(\sigma^2), R_{12}(\sigma^2), R_{21}(\sigma^2), R_{22}(\sigma^2)$ are defined in (\ref{four_rate}).
}
\end{cor}
The following result shows the condition for which CF without binning achieves the same rate region as noisy network coding.
\begin{thm}
\label{same_rate_region}
Compress-forward without binning achieves the same rate region as noisy network coding for the Gaussian TWRC if and only if
\begin{align}
\label{condition_2}
\sigma_{c1}^2&\leq\sigma_{e2}^2\nonumber\\
\sigma_{c2}^2&\leq\sigma_{e1}^2,
\end{align}
where $\sigma_{e1}^2,\sigma_{e2}^2$ are defined in (\ref{e1e2}) and $\sigma_{c1}^2,\sigma_{c2}^2$ are defined in (\ref{c1c2}). Otherwise, the rate region by compress-forward without binning is smaller.
\end{thm}
\begin{figure}[t]
\centering
  \includegraphics[scale=0.6]{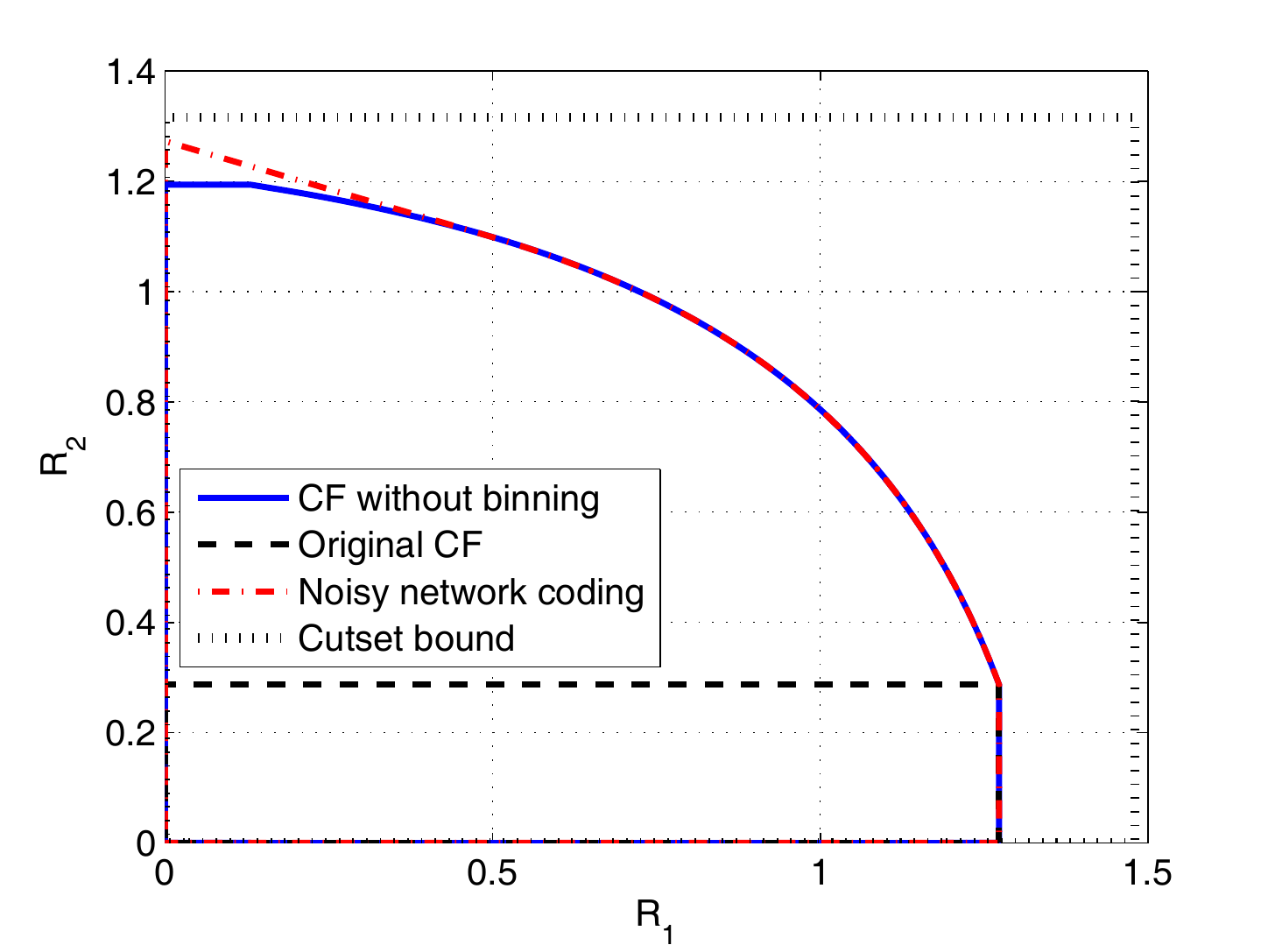}\\
  \caption{Rate regions for
    $P=20,g_{r1}=g_{1r}=2,g_{r2}=g_{2r}=0.5,g_{12}=g_{21}=0.1$.}\label{s1}
\end{figure}
\begin{figure}[t]
\centering
  \includegraphics[scale=0.6]{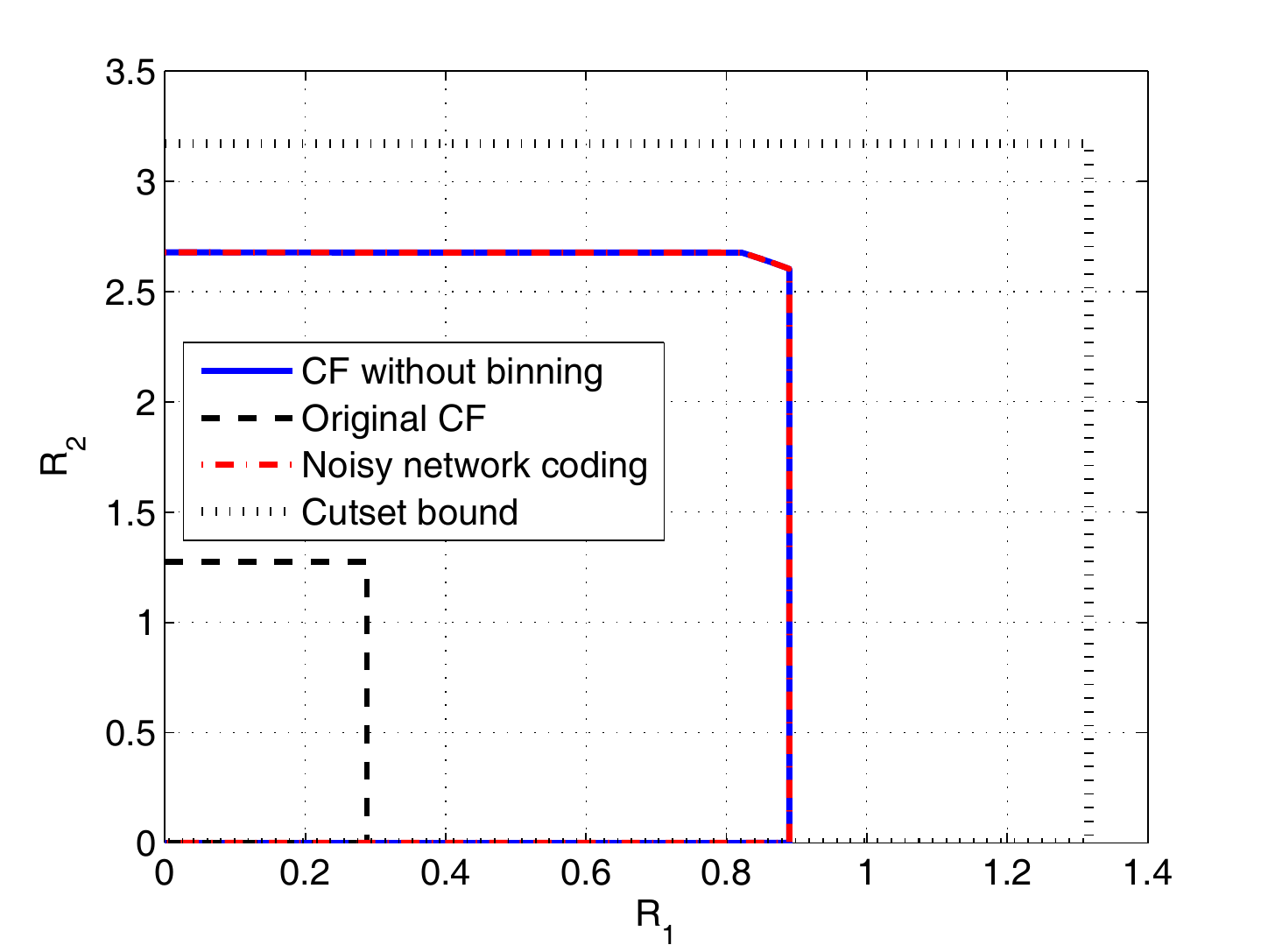}\\
  \caption{Rate regions for
    $P=20,g_{r1}=0.5,g_{1r}=2,g_{r2}=2,g_{2r}=0.5,g_{12}=g_{21}=0.1$.}\label{s5}
\end{figure}
\begin{proof}
Similar to the proof of Theorem \ref{thm:ori_equ_nobin}, note that both $R_{11}(\sigma^2),R_{21}(\sigma^2)$ are non-increasing and $R_{12}(\sigma^2),R_{22}(\sigma^2)$ are non-decreasing. Also,
\begin{align}
R_{11}(\sigma_{e1}^2)&=R_{12}(\sigma_{e1}^2)\nonumber\\
R_{21}(\sigma_{e2}^2)&=R_{22}(\sigma_{e2}^2).\nonumber
\end{align}
Therefore, the constraint in Theorem \ref{cor:G_nobin} is redundant if and only if
\begin{align}
\label{condition_1}
\max \{\sigma_{c1}^2,\sigma_{c2}^2\}\leq \min \{\sigma_{e1}^2,\sigma_{e2}^2\}.
\end{align}
Since $\sigma_{c1}^2\leq\sigma_{e1}^2$ and $\sigma_{c2}^2\leq\sigma_{e2}^2$ always hold, the above condition (\ref{condition_1}) is equivalent to (\ref{condition_2}).
\end{proof}\par

\begin{rem}
If condition \eqref{ori_equ_nobin} holds, then condition \eqref{condition_2} also holds. Thus when the original CF achieves the same rate region as CF without binning, it also achieves the same rate region as noisy network coding for the Gaussian TWRC.
\end{rem}
Figure \ref{s1} shows an asymmetric channel configuration in which CF without binning achieves a strictly larger rate region than the original CF scheme, but strictly smaller than noisy network coding. Figure \ref{s5} shows a case that CF without binning achieves the same rate region as noisy network coding and larger than CF with binning.\par

\subsection{Sum rate comparison with Noisy Network Coding}
We notice that in some cases, even though CF without binning achieves smaller rate region than noisy network coding, it still has the same sum rate. Thus we are interested in the channel conditions under which CF without binning and noisy network coding achieve the same sum rate.\par
Without loss of generality, assume $\sigma_{e1}^2 \geq \sigma_{e2}^2$ as in \eqref{e1e2}. First we analytically derive the optimal $\sigma^2$ that maximizes the sum rate of noisy network coding.
\begin{cor}
\label{cor:opti_sum_noisy}
Let $\sigma_{N}^2$ denotes the optimal $\sigma^2$ that maximizes the sum rate of the Gaussian TWRC using noisy network coding. Define $\sigma_{N1}^2$ and $\sigma_{N2}^2$ as follows:
\begin{align}
&\text{if}~\sigma_{e1}^2 \geq \sigma_{g}^2 \geq \sigma_{e2}^2,~\sigma_{N1}^2=\sigma_{g}^2;\nonumber\\
&\text{if}~\sigma_{g}^2 > \sigma_{e1}^2~\text{or}~\sigma_{g}^2 \leq 0,~\sigma_{N1}^2=\sigma_{e1}^2;\nonumber\\
\label{sigma_n1}
&\text{if}~\sigma_{e2}^2 > \sigma_{g}^2 > 0,~\sigma_{N1}^2=\sigma_{e2}^2;
\end{align}
and
\begin{align}
&\text{if}~\sigma_{e1}^2 \geq \sigma_{g}^2 \geq \sigma_{z1}^2,~\sigma_{N2}^2=\sigma_{g}^2;\nonumber \\
&\text{if}~\sigma_{g}^2 > \sigma_{e1}^2~\text{or}~\sigma_{g}^2 \leq 0,~\sigma_{N2}^2=\sigma_{e1}^2;\nonumber \\
&\text{if}~\sigma_{z1}^2 > \sigma_{g}^2 > 0,~\sigma_{N2}^2=\sigma_{z1}^2,
\label{sigma_n2}
\end{align}
where
\begin{align*}
\sigma_{g}^2&=\frac{g_{r2}^2P+g_{12}^2P+1}{g_{r2}^2P-g_{12}^2P-1},\\
\sigma_{z1}^2&=\frac{1}{g_{21}^2P+g_{2r}^2P}.
\end{align*}
Then,
\begin{align}
\text{if}~&\sigma_{z1}^2 \leq \sigma_{e2}^2,~\text{then~} \sigma_{N}^2=\sigma_{N1}^2;\nonumber \\
\text{if}~&\sigma_{z1}^2 > \sigma_{e2}^2,~\text{then}\nonumber\\
&\text{if}~R_{21}(\sigma_{e2}^2) \leq R_{12}(\sigma_{N2}^2)+R_{21}(\sigma_{N2}^2),~\sigma_{N}^2=\sigma_{N2}^2;\nonumber\\
&\text{if}~R_{21}(\sigma_{e2}^2) > R_{12}(\sigma_{N2}^2)+R_{21}(\sigma_{N2}^2),~\sigma_{N}^2=\sigma_{e2}^2.
\label{sigma_on2}
\end{align}
\end{cor}

\begin{proof}
See Appendix \ref{appendix:cor:opti_sum_noisy}.
\end{proof}

Based on Theorem \ref{cor:G_nobin} and Corollary \ref{cor:opti_sum_noisy}, we now obtain the following result on the conditions for CF without binning to achieve the same sum rate as noisy network coding for the Gaussian TWRC.
\begin{thm}
\label{same_sum_rate}
Compress-forward without binning achieves the same sum rate as noisy network coding for the Gaussian TWRC if and only if $\sigma_{c1}^2 \leq \sigma_{N}^2$, where $\sigma_{c1}^2$ is defined in (\ref{c1c2}) and $\sigma_{N}^2$ is defined in Corollary \ref{cor:opti_sum_noisy}. Otherwise, the sum rate achieved by compress-forward without binning is smaller.
\end{thm}

\begin{proof}
According to Corollary \ref{cor:opti_sum_noisy}, the sum rate of noisy network coding is maximized when $\sigma^2=\sigma_{N}^2$. According to Theorem \ref{cor:G_nobin}, for CF without binning, the constraint on $\sigma^2$ is $\sigma^2 \geq \sigma_{c1}^2$ (since we assume $\sigma_{c1}^2\geq \sigma_{c2}^2$). Therefore, compress-forward without binning achieves the same sum rate as noisy network coding if and only if the constraint region $\sigma^2 \geq \sigma_{c1}^2$ contains $\sigma_{N}^2$, which is equivalent to $\sigma_{c1}^2 \leq \sigma_{N}^2$.
\end{proof}

We can apply the above result to the special case of $g_{r1}=g_{1r}, g_{r2}=g_{2r}, g_{21}=g_{12}$ as follows.
\begin{figure}[t]
\centering
  \includegraphics[scale=0.6]{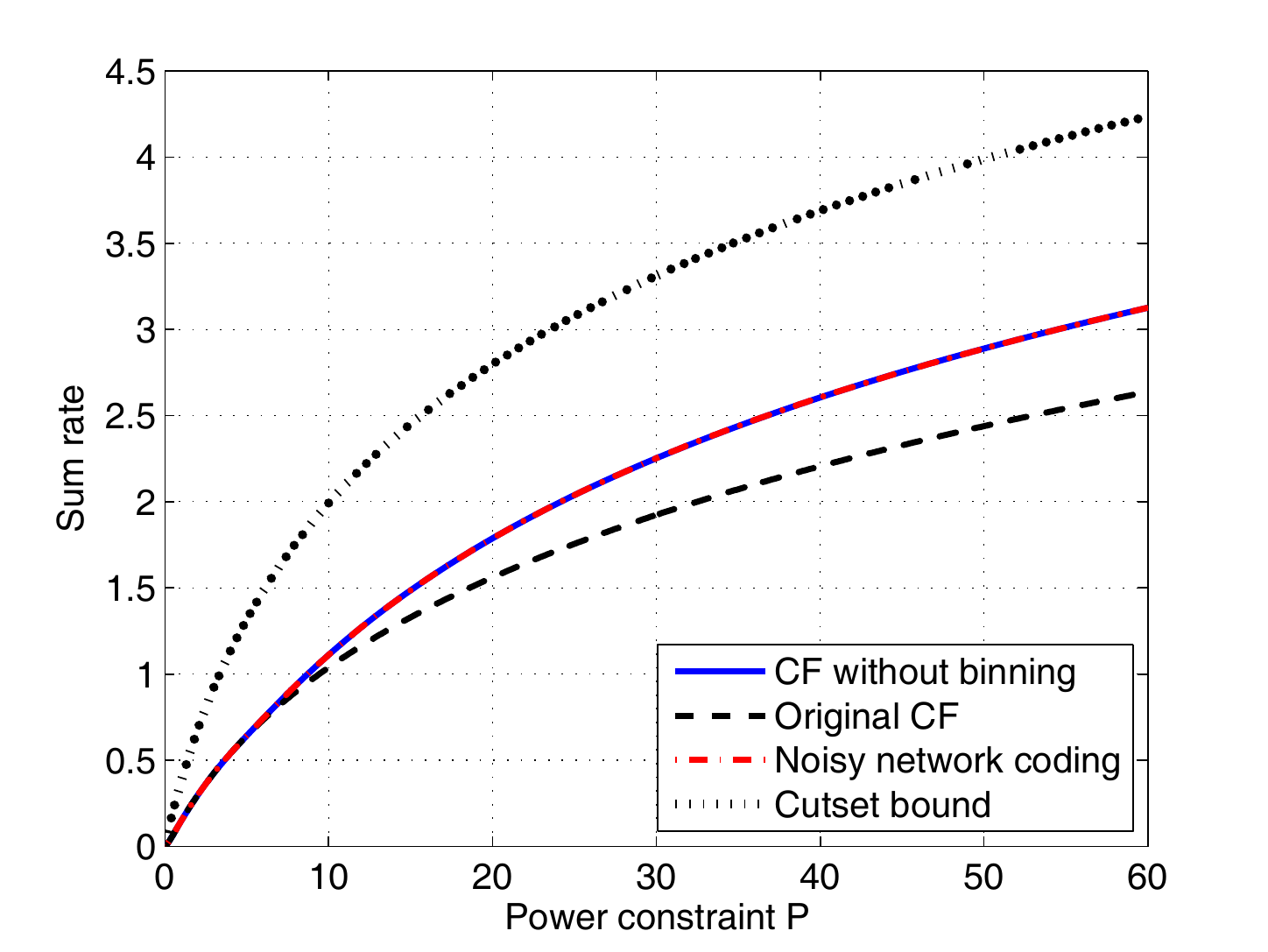}\\
  \caption{Sum rate for
    $g_{r1}=g_{1r}=2,g_{r2}=g_{2r}=0.5,g_{12}=g_{21}=0.1$.}\label{s2}
\end{figure}
\begin{cor}
\label{sumrate_allequ}
If $g_{r1}=g_{1r}, g_{r2}=g_{2r}, g_{21}=g_{12}$, then compress-forward without binning always achieves the same sum rate as noisy network coding.
\end{cor}
\begin{proof}
See Appendix \ref{appendix:cor:sumrate_allequ}.
\end{proof}
Figure \ref{s2} plots the sum rates for the same channel configurations as in Figure \ref{s1}, which shows the sum rates of CF without binning and noisy network coding are the same, even though the rate regions are not.\par
As confirmed in Figures \ref{s1}, \ref{s5} and \ref{s2}, CF without binning achieves a larger rate region and sum rate than the original CF scheme in \cite{rankov2006achievable} when the channel is asymmetric for the two users. CF without binning achieves the same rate region as noisy network coding when (\ref{condition_2}) is satisfied and achieves the same sum rate for a more relaxed condition. Furthermore, it has less decoding delay which is only 1 instead of $b$ blocks.\par

\section{Fading Two-way Relay Channel}\label{sec:FTWRC}
\label{sec:Fading_TWRC}
In this section, we derive the ergodic achievable rate regions for the fading two-way relay channel and compare them between CF without binning and noisy network coding. Similar to the Gaussian channel, we also derive the analytical conditions for when CF without binning achieves the same rate region or sum rate as noisy network coding.\par
The mathematical model of the fading TWRC channel can be expressed as in \eqref{GTWRC} with
\begin{align}
g_{12}=&\frac{h_{12}}{d_{12}^{\alpha/2}},\;\;\;g_{21}=\frac{h_{21}}{d_{12}^{\alpha/2}}\nonumber\\
g_{1r}=&\frac{h_{1r}}{d_{1r}^{\alpha/2}},\;\;\;g_{2r}=\frac{h_{2r}}{d_{2r}^{\alpha/2}}\nonumber\\
g_{r1}=&\frac{h_{r1}}{d_{1r}^{\alpha/2}},\;\;\;g_{r2}=\frac{h_{r2}}{d_{2r}^{\alpha/2}}\nonumber
\end{align}
\noindent where $d_{ij}$ is the distance between nodes $i$ and $j$ for $(i,j)\in \{1,2,r\}$, $\alpha$ is the pathloss attenuation exponent,  $h_{ij}$ are zero-mean complex Gaussian random variables (which corresponds to Rayleigh fading) and  $Z_i\sim {\cal CN}(0,1)$. Each $h_{ij}$ is independent from the other fading coefficients, all $X_i$ and all
$Z_i$.



\subsection{Ergodic achievable rate region}
An achievable rate region for the fading TWRC is obtained
by incorporating the randomness of the channel gains into Theorem \ref{cor:G_nobin}. The knowledge of fading coefficients available at each  node is as follows.
\begin{itemize}
  \item
  User 1 knows $h_{12},h_{1r},h_{r1}$ and $h_{r2}$. Hence, the effective output is $\tilde{Y}_1=[Y_1,h_{12},h_{1r},h_{r1},h_{r2}]$.
  \item
  User 2 knows $h_{21},h_{2r},h_{r1}$ and $h_{r2}$. Hence, $\tilde{Y}_2=[Y_2,h_{21},h_{2r},h_{r1},h_{r2}]$.
  \item
  The relay knows $h_{r1}$ and $h_{r2}$. Hence, $\tilde{Y}_r=[Y_r,h_{r1},h_{r2}]$.
\end{itemize}
In other words, each user knows the links from other nodes to itself and all links to the relay, while the relay only needs to know the links to itself.
We assume that each block of length $n$ is long
enough to experience all fading states such that the ergodic rate can be used. Denote
\begin{align}\label{fad1}
\bar{R}_{11}=&\;E\left\{C\left(\frac{|h_{21}|^2P}{d_{12}^\alpha}+\frac{|h_{r1}|^2P}{d_{1r}^\alpha(1+\sigma^2)}\right)\right\}\nonumber\\
\bar{R}_{12}=&\;E\left\{C\left(\frac{|h_{21}|^2P}{d_{12}^\alpha}+\frac{|h_{2r}|^2P}{d_{2r}^\alpha}\right)\right\}-C(1/\sigma^2)\nonumber\\
\bar{R}_{21}=&\;E\left\{C\left(\frac{|h_{12}|^2P}{d_{12}^\alpha}+\frac{|h_{r2}|^2P}{d_{2r}^\alpha(1+\sigma^2)}\right)\right\}\nonumber\\
\bar{R}_{22}=&\;E\left\{C\left(\frac{|h_{12}|^2P}{d_{12}^\alpha}+\frac{|h_{1r}|^2P}{d_{1r}^\alpha}\right)\right\}-C(1/\sigma^2)\nonumber\\
f_1=&\;E\left\{C\left(\frac{\frac{|h_{2r}|^2}{d_{2r}^\alpha}P}{1+\frac{|h_{21}|^2}{d_{12}^\alpha}P}\right)\right\}\nonumber\\
f_2=&\;E\left\{C\left(\frac{\frac{|h_{1r}|^2}{d_{1r}^\alpha}P}{1+\frac{|h_{12}|^2}{d_{12}^\alpha}P}\right)\right\}
\end{align}
\noindent where $C(|x|^2)=\log (1+|x|^2)$ and the expectations are taken over the fading parameters $h_{i,j}$.
Then, we obtain the following theorem:
\begin{thm}
\label{thm:fading_region}
The following ergodic rate region is achievable for the fading TWRC using CF without binning and joint decoding:
\begin{align}\label{fad5}
\bar{R}_1\leq&\;\min(\bar{R}_{11},\bar{R}_{12})\nonumber\\
\bar{R}_2\leq&\;\min(\bar{R}_{21},\bar{R}_{22})
\end{align}
%
for some
\begin{align}\label{consre}
\sigma^2\geq\;\max(\bar{\sigma}_{c1}^2,\bar{\sigma}_{c_2}^2)
\end{align}
where
\begin{align}
\bar{\sigma}_{c1}^2=&\frac{1}{2^{f_1}-1}\nonumber\\
\label{fading_c1_c2}
\bar{\sigma}_{c2}^2=&\frac{1}{2^{f_2}-1}.
\end{align}
and $f_1, f_2, \bar{R}_{11}, \bar{R}_{12}, \bar{R}_{21}, \bar{R}_{22}$ are defined in \eqref{fad1}.
\end{thm}

\begin{proof}
See Appendix \ref{pro8}.
\end{proof}

\begin{rem}
$\sigma^2$ in (\ref{consre}) is the same as that for the Gaussian channel because it results from the compression rate at the relay; however, $\bar{\sigma}_{c1}^2$ and $\bar{\sigma}_{c_2}^2$ are different from those for the Gaussian channel because they come from the decoding at each user and hence are affected by channel fading.
\end{rem}
The following Corollary evaluates the expectation of the rate constraints in Theorem \ref{thm:fading_region}.
\begin{cor}
\label{cor:fading_express}
Evaluating the expectation of the rate constraints in Theorem \ref{thm:fading_region} leads to the following expressions:
\begin{align}\label{fad2}
\bar{R}_{11}=&\;\frac{\lambda_{u_1}\lambda_{v_1}}{4\pi\sqrt{\pi}\ln{2}(\lambda_{v_1}-\lambda_{u_1})}
\times\;\left(G_{4,6}^{6,2}\left[\left(\frac{\lambda_{u_1}}{2}\right)^2\Bigg|_{b_q}^{a_p}\right]
-G_{4,6}^{6,2}\left[\left(\frac{\lambda_{v_1}}{2}\right)^2\Bigg|_{b_q}^{a_p}\right]\right)\\
\bar{R}_{12}=&\;-C(1/\sigma^2)+\frac{\lambda_{u_1}\lambda_{v_3}}{4\pi\sqrt{\pi}\ln{2}(\lambda_{v_3}-\lambda_{u_1})}
\times\;\left(G_{4,6}^{6,2}\left[\left(\frac{\lambda_{u_1}}{2}\right)^2\Bigg|_{b_q}^{a_p}\right]
-G_{4,6}^{6,2}\left[\left(\frac{\lambda_{v_3}}{2}\right)^2\Bigg|_{b_q}^{a_p}\right]\right)\nonumber\\
\bar{R}_{21}=&\;\frac{\lambda_{u_2}\lambda_{v_2}}{4\pi\sqrt{\pi}\ln{2}(\lambda_{v_2}-\lambda_{u_2})}
\times\;\left(G_{4,6}^{6,2}\left[\left(\frac{\lambda_{u_2}}{2}\right)^2\Bigg|_{b_q}^{a_p}\right]
-G_{4,6}^{6,2}\left[\left(\frac{\lambda_{v_2}}{2}\right)^2\Bigg|_{b_q}^{a_p}\right]\right)\nonumber\\
\bar{R}_{22}=&\;-C(1/\sigma^2)+\frac{\lambda_{u_2}\lambda_{v_4}}{4\pi\sqrt{\pi}\ln{2}(\lambda_{v_4}-\lambda_{u_2})}
\times\;\left(G_{4,6}^{6,2}\left[\left(\frac{\lambda_{u_2}}{2}\right)^2\Bigg|_{b_q}^{a_p}\right]
-G_{4,6}^{6,2}\left[\left(\frac{\lambda_{v_4}}{2}\right)^2\Bigg|_{b_q}^{a_p}\right]\right)\nonumber
\end{align}
where $a_p=[-0.5,0,0,0.5],$ $b_q=[0,0.5,-0.5,0,-0.5,0]$, $G[\cdot]$ is the Meijer's G-function \cite{tabofint}, and
\begin{align}
\lambda_{u_1}=&\;\frac{d_{12}^\alpha}{E[|h_{21}|^2]P},\;\;\lambda_{u_2}=\;\frac{d_{12}^\alpha}{E[|h_{12}|^2]P},\nonumber\\
\lambda_{v_3}=&\;\frac{d_{2r}^\alpha}{E[|h_{2r}|^2]P},\;\;\;\lambda_{v_4}=\;\frac{d_{1r}^\alpha}{E[|h_{1r}|^2]P},\nonumber\\
\lambda_{v_1}=&\;\frac{d_{1r}^\alpha(1+\sigma^2)}{E[|h_{r1}|^2]P},\;\;\lambda_{v_2}=\;\frac{d_{2r}^\alpha(1+\sigma^2)}{E[|h_{r2}|^2]P}.
\end{align}
\end{cor}

\begin{proof}
See Appendix \ref{pro9}.
\end{proof}

\subsection{Rate region comparison with Noisy Network Coding}
The following condition shows when noisy network coding and CF without binning achieve the same rate region for the fading TWRC.
\begin{cor}
Compress-forward without binning achieves the same rate region as noisy network region for the fading TWRC if and only if
\begin{align}\label{fad7}
\max\{\bar{\sigma}_{c1}^2,\bar{\sigma}_{c_2}^2\}\leq&\;\min\{\bar{\sigma}_{e1}^2,\bar{\sigma}_{e2}^2\},
\end{align}
where $\bar{\sigma}_{c1}^2,\bar{\sigma}_{c_2}^2$ come from \eqref{fading_c1_c2}; $\bar{\sigma}_{e1}^2$ is the intersection between $\bar{R}_{11}$ and $\bar{R}_{12}$, and $\bar{\sigma}_{e2}^2$ is the intersection between $\bar{R}_{21}$ and $\bar{R}_{22}$ as defined in \eqref{fad1}.
\end{cor}
\begin{proof}
Following the same analysis for the Gaussian channel in Theorem \ref{same_rate_region}, we can see from (\ref{fad1}) that $\bar{R}_{11}$ and $\bar{R}_{21}$ are non-increasing functions of $\sigma^2$ while $\bar{R}_{12}$ and $\bar{R}_{22}$ are non-decreasing functions. Therefore, the optimal $\bar{\sigma}_{e1}^2$ and $\bar{\sigma}_{e2}^2$ are those resulting from the intersections between $\bar{R}_{11},\bar{R}_{12}$ and $\bar{R}_{21},\bar{R}_{22}$, respectively.
\end{proof}

\subsection{Sum rate comparison with Noisy Network Coding}
Similar to the Gaussian channel case, we first analytically derive the optimal $\sigma^2$ that maximizes the sum rate of noisy network coding for the fading TWRC.
\begin{cor}
\label{cor:fading_sumrate}
The optimal $\bar{\sigma}_N^2$ that maximizes the sum rate of the fading TWRC using noisy network coding is obtained similarly to the Gaussian case
in Corollary \ref{cor:opti_sum_noisy} but with the following replacements:
\begin{itemize}
\item Replacing $\sigma_{z1}^2$ with $\bar{\sigma}_{z1}^2$ where
\begin{align}\label{sigz}
\bar{\sigma}_{z1}^2=&\;\frac{1}{2^{D_1}-1}\;\text{and}\\
D_1=&\;E\left\{C\left(\frac{|h_{21}|^2P}{d_{12}^\alpha}+\frac{|h_{2r}|^2P}{d_{2r}^\alpha}\right)\right\}\nonumber
\end{align}
\item Replacing $\sigma_{e1}, \sigma_{e2}$  with $\bar{\sigma}_{e1}, \bar{\sigma}_{e2}$ as in \eqref{fad7} respectively.
\item Replacing $\sigma_g$ with $\bar{\sigma}_g$ which is obtained by solving the following equation:
\begin{align}\label{define:sigma_g}
\frac{\partial}{\partial \bar{\sigma}}(\bar{R}_{12}+\bar{R}_{21})=0,
\end{align}
where $\bar{R}_{12}, \bar{R}_{21}$ are defined in \eqref{fad1}.
\end{itemize}
\end{cor}

\begin{proof}
The proof follows similar arguments as in the proof of Corollary \ref{cor:opti_sum_noisy}.
\end{proof}

\begin{rem}
$\bar{\sigma}_{z1}^2$ and $\bar{\sigma}_g^2$ in \eqref{sigz} and \eqref{define:sigma_g} can be evaluated using the Meijer G-function as shown in Appendix \ref{pro10}.
\end{rem}

Based on Corollary \ref{cor:fading_sumrate} and Theorem \ref{thm:fading_region}, we can now obtain the following theorem for the condition that CF without binning achieves the same sum rate as noisy network coding for the fading TWRC.
\begin{thm}
\label{thm:fading_sumrate}
Compress-forward without binning achieves the same sum rate as noisy network coding for the fading TWRC channel if and only if
$\bar{\sigma}_N^2>\max(\bar{\sigma}_{c1}^2,\bar{\sigma}_{c2}^2)$, where $\bar{\sigma}_N^2$ is obtained as in Corollary \ref{cor:fading_sumrate} and $\bar{\sigma}_{c1}^2,\bar{\sigma}_{c2}^2$ are defined in \eqref{fading_c1_c2}.
\end{thm}
\begin{proof}
Similar to the proof of Theorem \ref{same_sum_rate}.
\end{proof}

\subsection{Geometrical comparison}
In this section, we present the geometrical channel conditions for which CF without binning achieves the same rate region or sum rate as noisy network coding for both the Gaussian and fading two-way relay channels.\par
In the following discussions, the corresponding fading channel parameters ${g}_{ij}^2$ are related to the Gaussian channel parameters $\bar{g}_{ij}^2$ by the following equation:
\begin{align}
{g}_{ij}=h_{ij}\cdot \bar{g}_{ij}.
\end{align}
where $h_{ij} \sim \mathcal{N}(0,1)$ and $\bar{g}_{ij}=d_{ij}^{-\alpha/2}$ represents the pathloss components for $i,j\in\{1,2,r\}$.


\subsubsection{Equal pathloss ($\bar{g}_{r1}=\bar{g}_{1r}, \bar{g}_{r2}=\bar{g}_{2r}, \bar{g}_{21}=\bar{g}_{12}$)}
This channel configuration models large networks which consider the pathloss as a function of node locations.

%

Figure \ref{fig:case1} shows when CF without binning achieves the same rate region or sum rate as noisy network coding. We fix user 1 at point $(-0.5,0)$ and user 2 at point $(0.5,0)$. Let the pathloss exponent be $\alpha=2$.
For the Gaussian channel, Region 1 illustrates relay locations for which CF without binning achieves the same rate region as noisy network coding, while Region 2 (the whole plane) is for the same sum rate. Regions 3 and 4 are the corresponding ones for the fading channel. We see that in the equal pathloss case, CF without binning achieves the same sum rate as noisy network coding on the whole plane for Gaussian channels (as shown in Corollary \ref{sumrate_allequ}). Both the same-sum-rate and same-rate-region regions for the fading channel appear to be larger than those for the Gaussian channel.




\begin{figure}[t]
    \begin{center}
    \includegraphics[width=0.48\textwidth,height=75mm]{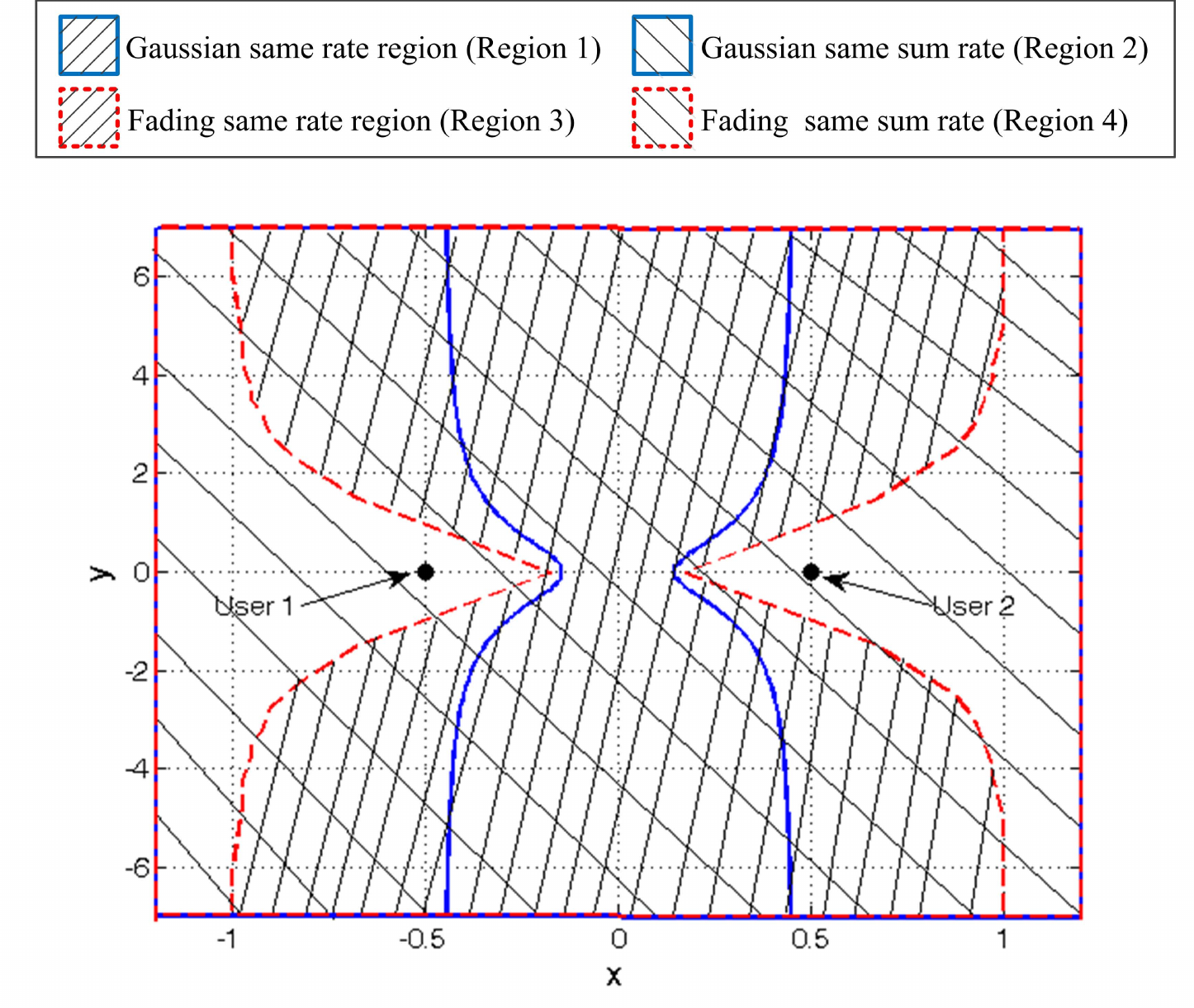}
    \caption{Locations of the relay for $\bar{g}_{r1}=\bar{g}_{1r}, \bar{g}_{r2}=\bar{g}_{2r}, \bar{g}_{21}=\bar{g}_{12}, P=10$ and $\alpha=2$.} \label{fig:case1}
    \end{center}
\end{figure}

\begin{figure}[t]
    \begin{center}
    \includegraphics[width=0.48\textwidth]{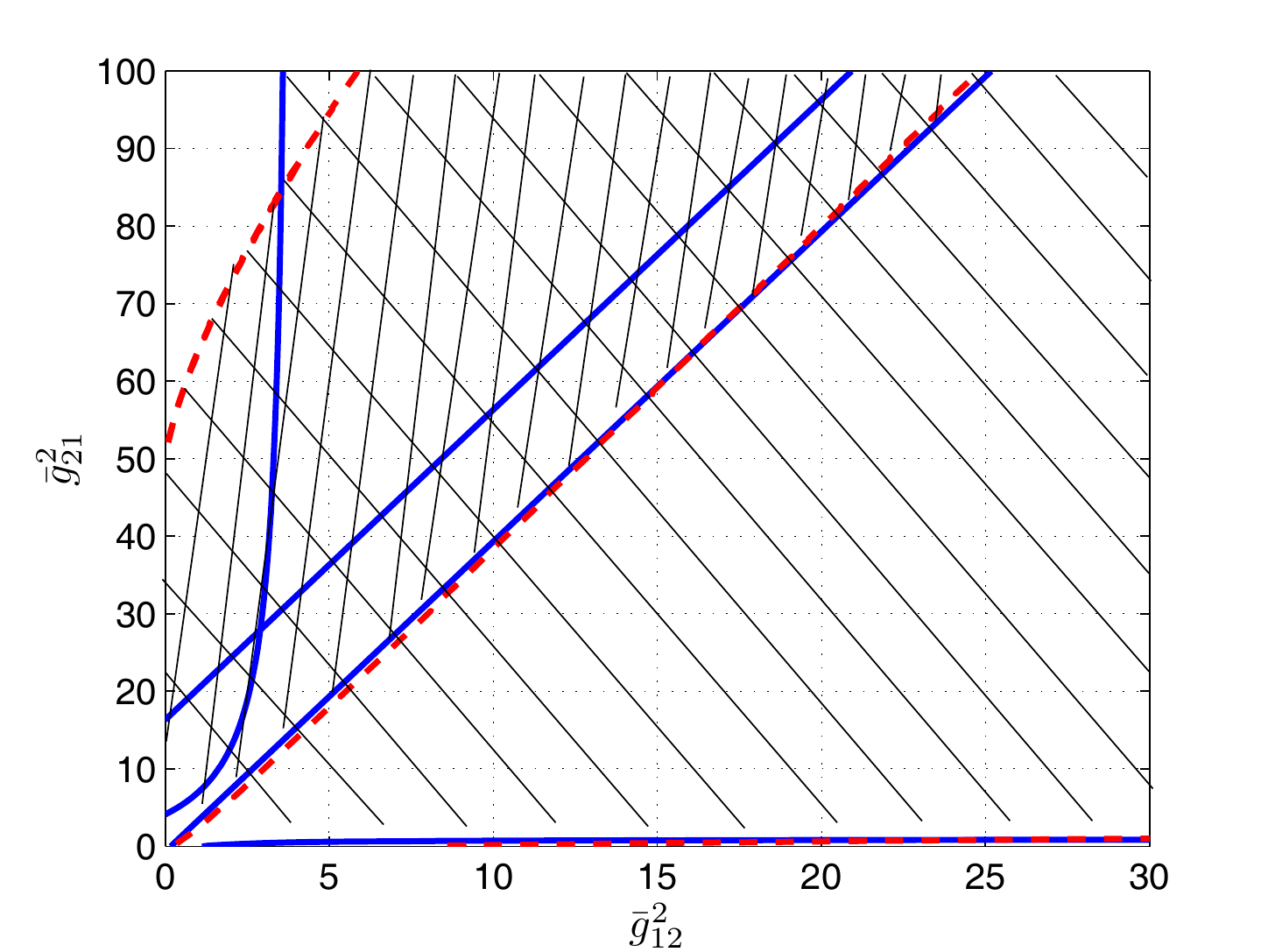}
    \caption{Region of $(\bar{g}_{12}^2,\bar{g}_{21}^2)$ for
    $\bar{g}_{r1}=\bar{g}_{1r}=1$ and $\bar{g}_{r2}=\bar{g}_{2r}=2$.} \label{fig:case2}
    \end{center}
\end{figure}

\begin{figure}[t]
    \begin{center}
    \includegraphics[width=0.48\textwidth]{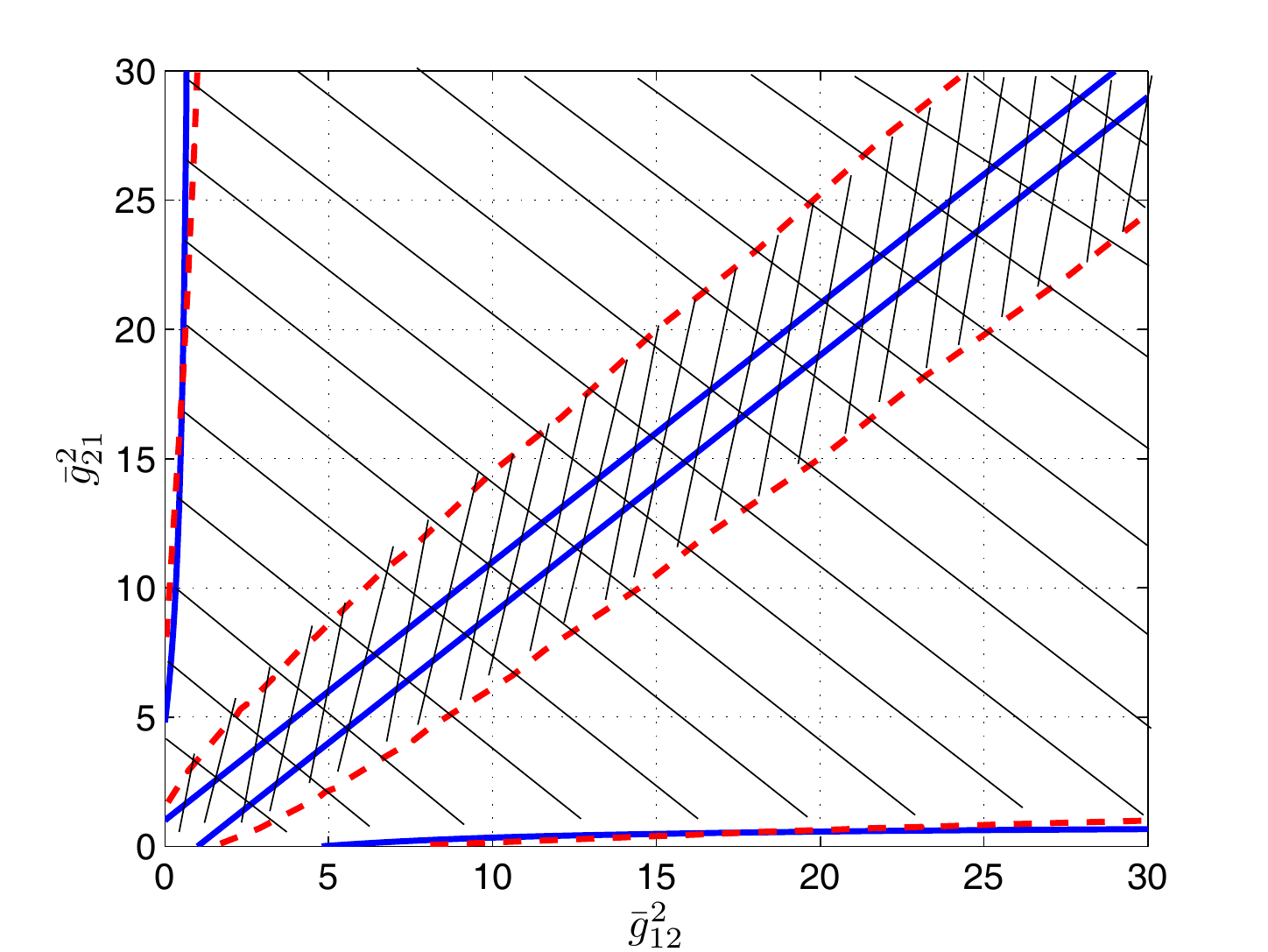}
    \caption{Region of $(\bar{g}_{12}^2,\bar{g}_{21}^2)$  for
    $\bar{g}_{r1}=\bar{g}_{r2}=1$ and $\bar{g}_{1r}=\bar{g}_{2r}=2$.}
    \label{fig:case3}
    \end{center}
\end{figure}

\subsubsection{Reciprocity ($\bar{g}_{r1}=\bar{g}_{1r}, \bar{g}_{r2}=\bar{g}_{2r}$)}
This channel applies when the nodes use time-division duplexing for which channel reciprocity holds.\par
Let $\bar{g}_{r1}=\bar{g}_{1r}=1, \bar{g}_{r2}=\bar{g}_{2r}=2$ and $P=10$. In Figure \ref{fig:case2}, Regions 1, 2, 3, 4 correspond to the same representations as those in Figure \ref{fig:case1}. In this case, CF without binning achieves the same sum rate as noisy network coding for the majority of $(\bar{g}_{12}^2,\bar{g}_{21}^2)$ for both Gaussian and fading channels.

\subsubsection{Uplink-downlink ($\bar{g}_{r1}=\bar{g}_{r2}, \bar{g}_{1r}=\bar{g}_{2r}$)}
This channel configuration applies if nodes use frequency-division duplexing, in which the uplink pathlosses can be similar and the downlinks can be similar.\par
Let $\bar{g}_{r1}=\bar{g}_{r2}=1, \bar{g}_{1r}=\bar{g}_{2r}=2$ and $P=10$. In Figure \ref{fig:case3}, Regions 1, 2, 3, 4 also correspond to the same representations as those in Figure \ref{fig:case1}. In this case, CF without binning also achieves the same sum rate as noisy network coding also for the majority of $(\bar{g}_{12}^2,\bar{g}_{21}^2)$ for both Gaussian and fading channels.  The region in which CF without binning achieves the same rate region as noisy network coding is larger for the fading channel than for the Gaussian channel.


\section{Conclusion}\label{Con}
We have analyzed impacts of the 3 new ideas for compress-forward (CF)
in noisy network coding (NNC): no Wyner-Ziv binning, relaxed
simultaneous decoding and message repetition. For the one-way relay
channel (single-source single-destination), no Wyner-Ziv binning alone
without message repetition can achieve the NNC rate at shorter
decoding delay. The extra requirement is joint (but not necessarily
relaxed) decoding at the destination of both the message and
compression index. However, for multi-source multi-destination
networks such as the two-way relay channel (TWRC), we find that all 3
techniques together is necessary to achieve the NNC rate region.\par

Under certain channel conditions, CF without Wyner-Ziv binning and
without message repetition can achieve the same rate region or sum
rate as NNC. Deriving these conditions explicitly for the Gaussian and
fading TWRC, we find that the difference in rate regions between CF
without binning and NNC is often small. Results also show that these
two schemes achieve the same sum rate for a majority of channel
configurations. The configurations for the same sum rates are more relaxed than for the
same rate regions and also are usually more relaxed for fading than for Gaussian channels.
These configurations may be useful in practice because
of the short encoding and decoding delay in CF without binning.

\appendices
\section{Proof of Theorem \ref{thm_rate}}
\label{appendix:thm_rate}
Assume without loss of generality that $m_j=1$ and $k_{j-1}=k_{j}=1$. First define the following two events:
\begin{align}
\mathcal{E}'_{1j}(k_j)&=\big\{(x_r^n(k_j),y^n(j+1))\in T^{(n)}_{\epsilon}\big\} \nonumber \\
\mathcal{E}'_{2j}(m_j,k_j)&=\big\{(x^n({m}_j),x_r^n(1),\hat{y}_r^n({k}_j|1),y^n(j))\in T^{(n)}_{\epsilon}\big\}. \nonumber
\end{align}
Then the decoder makes an error only if one or more of the following events occur:
\begin{align}
\mathcal{E}_{1j}=&\big\{(\hat{y}_r^n(k_j|1),y^n_r(j),x^n_r(1))\notin T^{(n)}_{\epsilon}~\textrm{for all}~k_j\in[1:2^{nR_r}]\big\} \nonumber \\
\mathcal{E}_{2j}=&\big\{(x_r^n(1),y^n(j+1))\notin T^{(n)}_{\epsilon}~\textrm{or}~(x^n(1),x_r^n(1),\hat{y}_r^n(1|1),y^n(j))\notin T^{(n)}_{\epsilon} \big\} \nonumber\\
\mathcal{E}_{3j}=&\big\{\mathcal{E}'_{1j}(k_j)~\textrm{and}~\mathcal{E}'_{2j}(1,k_j)~\textrm{for some}~k_j\neq 1\big\} \\
\mathcal{E}_{4j}=&\big\{\mathcal{E}'_{1j}(1)~\textrm{and}~\mathcal{E}'_{2j}(m_j,1)~\textrm{for some}~m_j\neq 1\big\} \nonumber \\
\mathcal{E}_{5j}=&\big\{\mathcal{E}'_{1j}(k_j)~\textrm{and}~\mathcal{E}'_{2j}(m_j,k_j)~\textrm{for some}~m_j\neq 1,k_j\neq 1\big\}.\nonumber
\end{align}
Thus, the probability of error is bounded as
\begin{align}
P\{\hat{m_j}\neq 1,\hat{k_j}\neq 1\}\leq P(\mathcal{E}_{1j})+P(\mathcal{E}_{2j}\cap\mathcal{E}_{1j}^c)+P(\mathcal{E}_{3j})+P(\mathcal{E}_{4j})+P(\mathcal{E}_{5j}). \nonumber
\end{align}
By the covering lemma \cite{gamal2010lecture}, $P(\mathcal{E}_{1j})\rightarrow 0$ as $n\rightarrow \infty$, if
\begin{align}
\label{rate_oneway_1}
R_r > I(\hat{Y}_r;Y_r|X_r).
\end{align}
By the conditional typicality lemma \cite{cover2006elements}, the second term $P(\mathcal{E}_{2j}\cap\mathcal{E}_{1j}^c) \rightarrow 0$ as $n\rightarrow \infty$.\par
For the rest of the error events, the decoded joint distribution for each event is as follows.
\begin{align*}
\mathcal{E}'_{1j}(k_j)&: p(x_r)p(y)\\
\mathcal{E}'_{2j}(1,k_j)&: p(x)p(x_r)p(\hat{y}_r|x_r)p(y|x,x_r)\\
\mathcal{E}'_{2j}(m_{j},1)&: p(x)p(x_r)p(y,\hat{y}_r|x_r)\\
\mathcal{E}'_{2j}(m_j,k_j)&: p(x)p(x_r)p(\hat{y}_r|x_r)p(y|x_r),
\end{align*}
where $m_j \neq 1, k_j \neq 1$. Using standard joint typicality analysis with the above decoded joint distribution, we can obtain a bound on each error event as follows.\par
Consider $\mathcal{E}_{3j}$:
\begin{align}
P(\mathcal{E}_{3j})&=P(\cup _{k_j\neq 1}(\mathcal{E}'_{1j}(k_j)\cap \mathcal{E}'_{2j}(1,k_j))) \nonumber \\
&\leq \sum _{k_j\neq 1} P(\mathcal{E}'_{1j}(k_j)) \times P(\mathcal{E}'_{2j}(1,k_j)). \nonumber
\end{align}
Note that if $k_j\neq 1$, then
\begin{align}
&P(\mathcal{E}'_{1j}(k_j))\leq 2^{-n(I(X_r;Y)-\delta(\epsilon))}; \nonumber \\
&P(\mathcal{E}'_{2j}(1,k_j))\nonumber \\
&= \sum _{(x,x_r,\hat{y}_r,y)\in T^{(n)}_{\epsilon}} p(x)p(x_r)p(\hat{y}_r|x_r)p(y|x,x_r) \nonumber \\
&\leq 2^{n(H(X,X_r,\hat{Y}_r,Y)-H(X)-H(X_r)-H(\hat{Y}_r|X_r)-H(Y|X,X_r)-3\delta(\epsilon))} \nonumber \\
&= 2^{n(H(\hat{Y}_r,Y|X,X_r)-H(\hat{Y}_r|X_r)-H(Y|X,X_r)-3\delta(\epsilon))} \nonumber \\
&=2^{n(H(\hat{Y}_r|Y,X,X_r)-H(\hat{Y}_r|X_r)-3\delta(\epsilon))} \nonumber \\
&=2^{-n(I(\hat{Y}_r;X,Y|X_r)-3\delta(\epsilon))}. \nonumber
\end{align}
Therefore
\begin{align}
P(\mathcal{E}_{3j})\leq 2^{nR_r}\cdot 2^{-n(I(X_r;Y)-\delta(\epsilon))} \cdot 2^{-n(I(\hat{Y}_r;X,Y|X_r)-3\delta(\epsilon))} \nonumber
\end{align}
which tends to zero as $n\rightarrow \infty $ if
\begin{align}
\label{rate_oneway_2}
R_r \leq I(X_r;Y)+I(\hat{Y}_r;X,Y|X_r).
\end{align}

Next consider $\mathcal{E}_{4j}$:
\begin{align}
P(\mathcal{E}_{4j})&=P(\cup _{m_{j}\neq 1}(\mathcal{E}'_{1j}(1)\cap \mathcal{E}'_{2j}(m_{j},1))) \nonumber \\
&\leq \sum _{m_{j}\neq 1} P(\mathcal{E}'_{2j}(m_{j},1)). \nonumber
\end{align}
Note that if $m_{j}\neq 1$, then
\begin{align}
&P(\mathcal{E}'_{2j}(m_{j},1))\nonumber \\
&= \sum _{(x,x_r,\hat{y}_r,y)\in T^{(n)}_{\epsilon}} p(x)p(x_r)p(y,\hat{y}_r|x_r) \nonumber \\
&\leq 2^{n(H(X,X_r,\hat{Y}_r,Y)-H(X)-H(X_r)-H(Y,\hat{Y}_r|X_r)-2\delta(\epsilon))} \nonumber \\
&= 2^{n(H(\hat{Y}_r,Y|X,X_r)-H(Y,\hat{Y}_r|X_r)-2\delta(\epsilon))} \nonumber \\
&=2^{-n(I(X;Y,\hat{Y}_r|X_r)-2\delta(\epsilon))}. \nonumber
\end{align}
Therefore
\begin{align}
P(\mathcal{E}_{4j})\leq &2^{nR}\cdot 2^{-n(I(X;Y,\hat{Y}_r|X_r)-2\delta(\epsilon))} \nonumber
\end{align}
which tends to zero as $n\rightarrow \infty $ if
\begin{align}
\label{rate_oneway_3}
R\leq I(X;Y,\hat{Y}_r|X_r).
\end{align}

Now consider $\mathcal{E}_{5j}$:
\begin{align}
P(\mathcal{E}_{5j})&=P(\cup _{m_j\neq 1} \cup _{k_j\neq 1}(\mathcal{E}'_{1j}(k_j)\cap \mathcal{E}'_{2j}(m_j,k_j))) \nonumber \\
&\leq \sum _{m_j\neq 1} \sum _{k_j\neq 1} P(\mathcal{E}'_{1j}(k_j)) \times P(\mathcal{E}'_{2j}(m_j,k_j)). \nonumber
\end{align}
Note that if $m_j\neq 1$ and $k_j\neq 1$, then
\begin{align}
&P(\mathcal{E}'_{1j}(k_j))\leq 2^{-n(I(X_r;Y)-\delta(\epsilon))}; \nonumber \\
&P(\mathcal{E}'_{2j}(m_j,k_j))\nonumber \\
&= \sum _{(x,x_r,\hat{y}_r,y)\in T^{(n)}_{\epsilon}} p(x)p(x_r)p(\hat{y}_r|x_r)p(y|x_r) \nonumber \\
&\leq 2^{n(H(X,X_r,\hat{Y}_r,Y)-H(X)-H(X_r)-H(\hat{Y}_r|X_r)-H(Y|X_r)-3\delta(\epsilon))} \nonumber \\
&= 2^{n(H(\hat{Y}_r,Y|X,X_r)-H(\hat{Y}_r|X_r)-H(Y|X_r)-3\delta(\epsilon))} \nonumber \\
&=2^{n(H(Y|X,X_r)+H(\hat{Y}_r|Y,X,X_r)-H(\hat{Y}_r|X_r)-H(Y|X_r)-3\delta(\epsilon))} \nonumber \\
&=2^{-n(I(X;Y|X_r)+I(\hat{Y}_r;X,Y|X_r)-3\delta(\epsilon))}. \nonumber
\end{align}
Therefore
\begin{align}
P(\mathcal{E}_{5j})\leq\;2^{nR}\cdot 2^{nR_r}\cdot 2^{-n(I(X_r;Y)-\delta(\epsilon))}\cdot 2^{-n(I(X;Y|X_r)+I(\hat{Y}_r;X,Y|X_r)-3\delta(\epsilon))} \nonumber
\end{align}
which tends to zero as $n\rightarrow \infty $ if
\begin{align}
\label{rate_oneway_4}
R+R_r &\leq I(X_r;Y)+I(X;Y|X_r)+I(\hat{Y}_r;X,Y|X_r) \nonumber \\
&= I(X,X_r;Y)+I(\hat{Y}_r;X,Y|X_r).
\end{align}

Combining the bounds (\ref{rate_oneway_1}) and (\ref{rate_oneway_4}), we have
\begin{align}
\label{rate_oneway_5}
R&\leq I(X,X_r;Y)+I(\hat{Y}_r;X,Y|X_r)-I(\hat{Y}_r;Y_r|X_r) \nonumber \\
&=I(X,X_r;Y)+I(\hat{Y}_r;X,Y|X_r)-I(\hat{Y}_r;Y_r,X,Y|X_r) \nonumber \\
&=I(X,X_r;Y)-I(\hat{Y}_r;Y_r|X,X_r,Y).
\end{align}
Combining (\ref{rate_oneway_1}), (\ref{rate_oneway_2}) and (\ref{rate_oneway_3}), (\ref{rate_oneway_5}), we obtain the result of Theorem \ref{thm_rate}.

\setcounter{subsubsection}{0}
\section{Proof of Theorem \ref{rate_two_way_nobin}}
\label{appendix:rate_two_way_nobin}
Assume without loss of generality that $m_{1,j}=m_{1,j+1}=m_{2,j}=1$ and $k_{j-1}=k_{j}=1$. First define the following two events:
\begin{align}
\mathcal{E}'_{1j}(k_j)&=\big\{(x_r^n(k_j),y_1^n(j+1),x_1^n(1))\in T^{(n)}_{\epsilon}\big\} \nonumber \\
\mathcal{E}'_{2j}(m_{2,j},k_j)&=\big\{(x_2^n({m}_{2,j}),x_r^n(1),\hat{y}_r^n({k}_j|1),y_1^n(j),x_1^n(1))\in T^{(n)}_{\epsilon}\big\}. \nonumber
\end{align}
Then the decoder makes an error only if one or more of the following events occur:
\begin{align}
\mathcal{E}_{1j}=&\big\{(\hat{y}_r^n(k_j|1),y^n_r(j),x^n_r(1))\notin T^{(n)}_{\epsilon}~\textrm{for all}~k_j\in[1:2^{nR_r}]\big\} \nonumber \\
\mathcal{E}_{2j}=&\big\{(x_r^n(1),y_1^n(j+1),x_1^n(1))\notin T^{(n)}_{\epsilon}~\textrm{or}~ (x_2^n(1),x_r^n(1),\hat{y}_r^n(1|1),y_1^n(j),x_1^n(1))\notin T^{(n)}_{\epsilon} \big\} \nonumber\\
\label{error:two-way_nobin}
\mathcal{E}_{3j}=&\big\{\mathcal{E}'_{1j}(k_j)~\textrm{and}~\mathcal{E}'_{2j}(1,k_j)~\textrm{for some}~k_j\neq 1\big\} \\
\mathcal{E}_{4j}=&\big\{\mathcal{E}'_{1j}(1)~\textrm{and}~\mathcal{E}'_{2j}(m_{2,j},1)~\textrm{for some}~m_{2,j}\neq 1\big\} \nonumber \\
\mathcal{E}_{5j}=&\big\{\mathcal{E}'_{1j}(k_j)~\textrm{and}~\mathcal{E}'_{2j}(m_{2,j},k_j)~\textrm{for some}~m_{2,j}\neq 1,k_j\neq 1\big\}.\nonumber
\end{align}
Thus, the probability of error is bounded as
\begin{align}
P\{\hat{m}_{2,j}\neq 1,\hat{k_j}\neq 1\}\leq & P(\mathcal{E}_{1j})+P(\mathcal{E}_{2j}\cap\mathcal{E}_{1j}^c)+P(\mathcal{E}_{3j})+P(\mathcal{E}_{4j})+P(\mathcal{E}_{5j}). \nonumber
\end{align}
By the covering lemma, $P(\mathcal{E}_{1j}) \rightarrow 0$ as $n\rightarrow \infty$, if
\begin{align}
\label{rate_1}
R_r > I(\hat{Y}_r;Y_r|X_r).
\end{align}
By the conditional typicality lemma, the second term $P(\mathcal{E}_{2j}\cap\mathcal{E}_{1j}^c)\rightarrow 0$ as $n\rightarrow \infty$. \par
For the rest of the error events, the decoded joint distribution for each event is as follows.
\begin{align}
\mathcal{E}'_{1j}(k_j)&: p(x_1)p(x_r)p(y_1|x_1)\nonumber\\
\mathcal{E}'_{2j}(1,k_j)&: p(x_1)p(x_2)p(x_r)p(\hat{y}_r|x_r)p(y_1|x_2,x_r,x_1)\nonumber\\
\mathcal{E}'_{2j}(m_{2,j},1)&: p(x_1)p(x_2)p(x_r)p(y_1,\hat{y}_r|x_r,x_1)\nonumber\\
\mathcal{E}'_{2j}(m_{2,j},k_j)&: p(x_1)p(x_2)p(x_r)p(\hat{y}_r|x_r)p(y_1|x_r,x_1),
\label{joint_distribution_1}
\end{align}
where $m_{2,j} \neq 1, k_j \neq 1$. Using standard joint typicality analysis with the above decoded joint distribution, we can obtain a bound on each error event as follows.\par
Consider $\mathcal{E}_{3j}$:
\begin{align}
P(\mathcal{E}_{3j})&=P(\cup _{k_j\neq 1}(\mathcal{E}'_{1j}(k_j)\cap \mathcal{E}'_{2j}(1,k_j))) \nonumber \\
&\leq \sum _{k_j\neq 1} P(\mathcal{E}'_{1j}(k_j)) \times P(\mathcal{E}'_{2j}(1,k_j)). \nonumber
\end{align}
Note that if $k_j\neq 1$, then
\begin{align}
&P(\mathcal{E}'_{1j}(k_j))\leq 2^{-n(I(X_r;Y_1|X_1)-\delta(\epsilon))}; \nonumber \\
&P(\mathcal{E}'_{2j}(1,k_j))\nonumber \\
&= \sum _{(x_1,x_2,x_r,\hat{y}_r,y_1)\in T^{(n)}_{\epsilon}} p(x_1)p(x_2)p(x_r)p(\hat{y}_r|x_r)p(y_1|x_2,x_r,x_1)\nonumber\\
&\leq 2^{n(H(X_1,X_2,X_r,\hat{Y}_r,Y_1)-H(X_1)-H(X_2)-H(X_r)-H(\hat{Y}_r|X_r)-H(Y_1|X_2,X_r,X_1)-4\delta(\epsilon))}\nonumber\\
&= 2^{n(H(\hat{Y}_r,Y_1|X_2,X_r,X_1)-H(\hat{Y}_r|X_r)-H(Y_1|X_2,X_r,X_1)-4\delta(\epsilon))} \nonumber \\
&=2^{n(H(\hat{Y}_r|Y_1,X_2,X_r,X_1)-H(\hat{Y}_r|X_r)-4\delta(\epsilon))} \nonumber \\
&=2^{-n(I(\hat{Y}_r;X_1,X_2,Y_1|X_r)-4\delta(\epsilon))}.\nonumber
\end{align}
Therefore
\begin{align}
P(\mathcal{E}_{3j})&\leq 2^{nR_r}\cdot 2^{-n(I(X_r;Y_1|X_1)-\delta(\epsilon))}\cdot 2^{-n(I(\hat{Y}_r;X_1,X_2,Y_1|X_r)-4\delta(\epsilon))} \nonumber
\end{align}
which tends to zero as $n\rightarrow \infty $ if
\begin{align}
\label{rate_Rr}
R_r \leq I(X_r;Y_1|X_1)+I(\hat{Y}_r;X_1,X_2,Y_1|X_r).
\end{align}

Next consider $\mathcal{E}_{4j}$:
\begin{align}
P(\mathcal{E}_{4j})&=P(\cup _{m_{2,j}\neq 1}(\mathcal{E}'_{1j}(1)\cap \mathcal{E}'_{2j}(m_{2,j},1))) \nonumber \\
&\leq \sum _{m_{2,j}\neq 1} P(\mathcal{E}'_{2j}(m_{2,j},1)). \nonumber
\end{align}
Note that if $m_{2,j}\neq 1$, then
\begin{align}
&P(\mathcal{E}'_{2j}(m_{2,j},1))\nonumber \\
&= \sum _{(x_1,x_2,x_r,\hat{y}_r,y_1)\in T^{(n)}_{\epsilon}} p(x_1)p(x_2)p(x_r)p(y_1,\hat{y}_r|x_r,x_1) \nonumber \\
&\leq 2^{n(H(X_1,X_2,X_r,\hat{Y}_r,Y_1)-H(X_1)-H(X_2)-H(X_r)-H(Y_1,\hat{Y}_r|X_r,X_1)-3\delta(\epsilon))}\nonumber\\
&= 2^{n(H(\hat{Y}_r,Y_1|X_2,X_r,X_1)-H(Y_1,\hat{Y}_r|X_r,X_1)-3\delta(\epsilon))} \nonumber \\
&=2^{-n(I(X_2;Y_1,\hat{Y}_r|X_r,X_1)-3\delta(\epsilon))}. \nonumber
\end{align}
Therefore
\begin{align}
P(\mathcal{E}_{4j})\leq &2^{nR_2}\cdot 2^{-n(I(X_2;Y_1,\hat{Y}_r|X_r,X_1)-3\delta(\epsilon))} \nonumber
\end{align}
which tends to zero as $n\rightarrow \infty $ if
\begin{align}
\label{rate_twoway_5}
R_2\leq I(X_2;Y_1,\hat{Y}_r|X_r,X_1).
\end{align}

Now consider $\mathcal{E}_{5j}$:
\begin{align}
P(\mathcal{E}_{5j})&=P(\cup _{m_{2,j}\neq 1} \cup _{k_j\neq 1}(\mathcal{E}'_{1j}(k_j)\cap \mathcal{E}'_{2j}(m_{2,j},k_j))) \nonumber \\
&\leq \sum _{m_{2,j}\neq 1} \sum _{k_j\neq 1} P(\mathcal{E}'_{1j}(k_j)) \times P(\mathcal{E}'_{2j}(m_{2,j},k_j)). \nonumber
\end{align}
Note that if $m_{2,j}\neq 1$ and $k_j\neq 1$, then
\begin{align}
&P(\mathcal{E}'_{1j}(k_j))\leq 2^{-n(I(X_r;Y_1|X_1)-\delta(\epsilon))}; \nonumber \\
&P(\mathcal{E}'_{2j}(m_{2,j},k_j))\nonumber \\
&= \sum _{(x_1,x_2,x_r,\hat{y}_r,y_1)\in T^{(n)}_{\epsilon}} p(x_1)p(x_2)p(x_r)p(\hat{y}_r|x_r)p(y_1|x_r,x_1) \nonumber \\
&\leq 2^{n(H(X_1,X_2,X_r,\hat{Y}_r,Y_1)-H(X_1)-H(X_2)-H(X_r)-H(\hat{Y}_r|X_r)-H(Y_1|X_r,X_1)-4\delta(\epsilon))}\nonumber \\
&= 2^{n(H(\hat{Y}_r,Y_1|X_2,X_r,X_1)-H(\hat{Y}_r|X_r)-H(Y_1|X_r,X_1)-4\delta(\epsilon))} \nonumber \\
&=2^{n(H(Y_1|X_2,X_r,X_1)+H(\hat{Y}_r|Y_1,X_2,X_r,X_1)-H(\hat{Y}_r|X_r)-H(Y_1|X_r,X_1)-4\delta(\epsilon))}\nonumber \\
&=2^{-n(I(X_2;Y_1|X_1,X_r)+I(\hat{Y}_r;X_1,X_2,Y_1|X_r)-4\delta(\epsilon))}. \nonumber
\end{align}
Therefore
\begin{align}
P(\mathcal{E}_{5j})\leq 2^{nR_2}\cdot 2^{nR_r}\cdot 2^{-n(I(X_r;Y_1|X_1)-\delta(\epsilon))}\cdot 2^{-n(I(X_2;Y_1|X_1,X_r)+I(\hat{Y}_r;X_1,X_2,Y_1|X_r)-4\delta(\epsilon))} \nonumber
\end{align}
which tends to zero as $n\rightarrow \infty $ if
\begin{align}
\label{rate_4}
R_2+R_r &\leq I(X_r;Y_1|X_1)+I(X_2;Y_1|X_1,X_r)
+I(\hat{Y}_r;X_1,X_2,Y_1|X_r) \nonumber \\
&= I(X_2,X_r;Y_1|X_1)+I(\hat{Y}_r;X_1,X_2,Y_1|X_r).
\end{align}

Combining the bounds (\ref{rate_1}) and (\ref{rate_4}), we have
\begin{align}
\label{rate_7}
R_2 &\leq I(X_2,X_r;Y_1|X_1)+I(\hat{Y}_r;X_1,X_2,Y_1|X_r)
-I(\hat{Y}_r;Y_r|X_r)\nonumber\\
&=I(X_2,X_r;Y_1|X_1)-I(\hat{Y}_r;Y_r|X_1,X_2,X_r,Y_1).
\end{align}
Combining the bounds (\ref{rate_twoway_5}) and (\ref{rate_7}), we obtain the rate constraint on $R_2$ in Theorem \ref{rate_two_way_nobin}. Similar for $R_1$. From (\ref{rate_1}) and (\ref{rate_Rr}), we obtain constraint (\ref{const_1}).

\setcounter{subsubsection}{0}
\section{Proof of Theorem \ref{thm_rate_ori}}
\label{appendix:thm_rate_ori}
We use a block coding scheme in which each user sends $b-1$ messages over $b$ blocks of $n$ symbols each.\par
\subsubsection{Codebook generation}
 Fix $p(x_1)p(x_2)p(x_r)p(\hat{y}_r|x_r,y_r)$. We randomly and independently generate a codebook for each block $j\in [1:b]$
\begin{itemize}
\item Independently generate $2^{nR_1}$ sequences $x_1^n(m_{1,j})\sim\prod ^n_{i=1}p(x_{1i}) $, where $m_{1,j} \in [1:2^{nR_1}]$.
\item Independently generate $2^{nR_2}$ sequences $x_2^n(m_{2,j})\sim\prod ^n_{i=1}p(x_{2i}) $, where $m_{2,j} \in [1:2^{nR_2}]$.
\item Independently generate $2^{nR_r}$ sequences $x_r^n(q_{j-1})\sim\prod ^n_{i=1}p(x_{ri})$, where $q_{j-1} \in [1:2^{nR_r}]$.
\item For each $q_{j-1}\in [1:2^{nR_r}]$, independently generate $2^{n(R_r+R'_r)}$  sequences $\hat{y}^n_r(q_j,r_j|q_{j-1})\sim\prod ^n_{i=1}p(\hat{y}_{ri}|x_{ri}(q_{j-1}))$. Throw them into $2^{nR_r}$ bins, where $q_j\in [1:2^{nR_r}]$ denotes the bin index and $r_j\in [1:2^{nR'_r}]$ denotes the relative index within a bin.
\end{itemize}

\subsubsection{Encoding}
 User 1 and user 2 transmits $x_1^n(m_{1,j})$ and $x_2^n(m_{2,j})$ in block $j$ separately. The relay, upon receiving $y^n_r(j)$, finds an index pair $(q_j,r_j)$ such that
 \begin{align*}
(\hat{y}_r^n(q_j,r_j|q_{j-1}),y^n_r(j),x^n_r(q_{j-1}))\in T^{(n)}_{\epsilon}.
\end{align*}
Assume that such $(q_j,r_j)$ is found, the relay sends $x^n_r(q_j)$ in block $j+1$. By the covering lemma, the probability that there is no such $(q_j,r_j)$ tends to 0 as $n\rightarrow \infty$ if
\begin{align}
\label{rate_twoway_ori_1}
R_r+R'_r > I(\hat{Y}_r;Y_r|X_r).
\end{align}

\subsubsection{Decoding}
At the end of block $j$, user 1 determines the unique $\hat{q}_{j-1}$ such that
\begin{align*}
(x^n_r(\hat{q}_{j-1}),y^n_1(j),x^n_1(m_{1,j})) \in T^{(n)}_{\epsilon}.
\end{align*}
Similar for user 2. Both succeed with high probability if
\begin{align}
\label{rate_twoway_ori_2}
R_r\leq \min \{I(X_r;Y_1|X_1),I(X_r;Y_2|X_2)\}.
\end{align}
Then user 1 uses $y^n_1(j-1)$ to determine the unique $\hat{r}_{j-1}$ such that
\begin{align*}
(\hat{y}^n_r(\hat{q}_{j-1},\hat{r}_{j-1}|\hat{q}_{j-2}),y^n_1(j-1),x^n_r(\hat{q}_{j-2}),x^n_1(m_{1,j-1}))\in T^{(n)}_{\epsilon}.
\end{align*}
Similar for user 2. Both succeed with high probability if
\begin{align}
\label{rate_twoway_ori_3}
R'_r\leq \min \{I(\hat{Y}_r;X_1,Y_1|X_r),I(\hat{Y}_r;X_2,Y_2|X_r)\}.
\end{align}
Finally, user 1 uses both $y^n_1(j-1)$ and $\hat{y}^n_r(j-1)$ to determine the unique $\hat{m}_{2,j-1}$ such that
\begin{align*}
(x^n_2(\hat{m}_{2,j-1}),\hat{y}^n_r(\hat{q}_{j-1},\hat{r}_{j-1}|\hat{q}_{j-2}),y^n_1(j-1),x^n_r(\hat{q}_{j-2}), x^n_1(m_{1,j-1}))\in T^{(n)}_{\epsilon}.
\end{align*}
Similar for user 2. Both succeed with high probability if
\begin{align}
R_1& \leq I(X_1;Y_2,\hat{Y}_r|X_2,X_r)\nonumber\\
\label{rate_twoway_ori_5}
R_2& \leq I(X_2;Y_1,\hat{Y}_r|X_1,X_r).
\end{align}
The constraint (\ref{const_2}) comes from combining (\ref{rate_twoway_ori_1}), (\ref{rate_twoway_ori_2}) and (\ref{rate_twoway_ori_3}).

\begin{rem}
\label{rem:error}
We note an error in the proof in \cite{rankov2006achievable}. In \cite{rankov2006achievable}, when user 1 determines the unique $\hat{r}_{j-1}$, it is stated that it succeeds with high probability if
\begin{align}
\label{wrong}
R'_r\leq I(\hat{Y}_r;Y_1|X_1,X_r),
\end{align}
which corresponds to step (\ref{rate_twoway_ori_3}) in our analysis. However, this is incorrect since the decoded joint distribution of this error event is $p(x_1)p(x_r)p(\hat{y}_r|x_r)p(y_1|x_1,x_r)$. Therefore, the error probability can be bounded as
\begin{align*}
&P(\mathcal{E})=\sum _{(x_1,x_r,\hat{y}_r,y_1)\in T^{(n)}_{\epsilon}} p(x_1)p(x_r)p(\hat{y}_r|x_r)p(y_1|x_1,x_r)\\
&\leq 2^{n(H(X_1,X_r,\hat{Y}_r,Y_1)-H(X_1)-H(X_r)-H(\hat{Y}_r|X_r)-H(Y_1|X_1,X_r)-3\delta(\epsilon))}\\
&=2^{-n(I(\hat{Y}_r;X_1,Y_1|X_r)-3\delta(\epsilon))},
\end{align*}
which tends to zero as as $n\rightarrow \infty $ if (\ref{rate_twoway_ori_3}) is satisfied instead of (\ref{wrong}).
\end{rem}

\setcounter{subsubsection}{0}
\section{Proof of Corollary \ref{cor:no_index_decoding}}
\label{appendix:cor:no_index_decoding}
The codebook generation and encoding is the same as that in Theorem \ref{rate_two_way_nobin}.
The decoding rule is changed as follows: in block $j+1$, user 1 finds a unique $\hat{m}_{2,j}$ such that
\begin{align}
(x_2^n(\hat{m}_{2,j}),x_r^n(\hat{k}_{j-1}),\hat{y}_r^n(\hat{k}_j|\hat{k}_{j-1}),y_1^n(j),x_1^n(m_{1,j}))&\in T^{(n)}_{\epsilon} \nonumber\\
\textrm{and}~~~~~~~~~~~~~~(x_r^n(\hat{k}_j),y_1^n(j+1),x_1^n(m_{1,j+1}))&\in T^{(n)}_{\epsilon} \nonumber
\end{align}
for some pair of $(\hat{k}_{j-1},\hat{k}_{j})$.
Next we present the error analysis.\par
Assume without loss of generality that $m_{1,j}=m_{1,j+1}=m_{2,j}=1$ and $k_{j-1}=k_{j}=1$. First define the following two events:
\begin{align}
\mathcal{E}'_{1j}(k_j)&=\big\{(x_r^n(k_j),y_1^n(j+1),x_1^n(1))\in T^{(n)}_{\epsilon}\big\} \nonumber \\
\mathcal{E}'_{2j}(m_{2,j},k_j,k_{j-1})&=\big\{(x_2^n({m}_{2,j}),x_r^n(k_{j-1}),\hat{y}_r^n({k}_j|k_{j-1}),y_1^n(j),x_1^n(1))\in T^{(n)}_{\epsilon}\big\}. \nonumber
\end{align}
Then the decoder makes an error only if one or more of the following events occur:

\begin{align}
\mathcal{E}_{1j}=&\big\{(\hat{y}_r^n(k_j|1),y^n_r(j),x^n_r(1))\notin T^{(n)}_{\epsilon}~\textrm{for all}~k_j\in[1:2^{nR_r}]\big\} \nonumber \\
\mathcal{E}_{2j}=&\big\{(x_r^n(1),y_1^n(j+1),x_1^n(1))\notin T^{(n)}_{\epsilon}~\textrm{or}~(x_2^n(1),x_r^n(1),\hat{y}_r^n(1|1),y_1^n(j),x_1^n(1))\notin T^{(n)}_{\epsilon} \big\} \nonumber\\
\tilde{\mathcal{E}}_{4j}=&\big\{\mathcal{E}'_{1j}(1)~\textrm{and}~\mathcal{E}'_{2j}(m_{2,j},1,1)~\textrm{for some}~m_{2,j}\neq 1\big\} \nonumber \\
\tilde{\mathcal{E}}_{5j}=&\big\{\mathcal{E}'_{1j}(k_j)~\textrm{and}~\mathcal{E}'_{2j}(m_{2,j},k_j,1)~\textrm{for some}~m_{2,j}\neq 1,k_j\neq 1\big\} \nonumber\\
\mathcal{E}_{6j}=&\big\{\mathcal{E}'_{1j}(1)~\textrm{and}~\mathcal{E}'_{2j}(m_{2,j},1,k_{j-1})\textrm{~for some}
~m_{2,j}\neq 1,k_{j-1}\neq 1\big\}\nonumber\\
\mathcal{E}_{7j}=&\big\{\mathcal{E}'_{1j}(k_j)~\textrm{and}~\mathcal{E}'_{2j}(m_{2,j},k_j,k_{j-1})~\textrm{for some}
~m_{2,j}\neq 1,k_j\neq 1,k_{j-1}\neq 1\big\}, \nonumber
\end{align}
where $\mathcal{E}_{1j}, \mathcal{E}_{2j}$ are the same as in \eqref{error:two-way_nobin} in the proof of Theorem \ref{rate_two_way_nobin} (Appendix \ref{appendix:rate_two_way_nobin}); $\mathcal{E}_{3j}$ in \eqref{error:two-way_nobin} is not an error here. $\tilde{\mathcal{E}}_{4j}$ and $\tilde{\mathcal{E}}_{5j}$ are similar to ${\mathcal{E}}_{4j}$ and ${\mathcal{E}}_{5j}$ in \eqref{error:two-way_nobin}. $\mathcal{E}_{6j}$ and $\mathcal{E}_{7j}$ are new error events.\par
The probability of error is bounded as
\begin{align}
P\{\hat{m}_{2,j}\neq 1\} \leq  P(\mathcal{E}_{1j})+P(\mathcal{E}_{2j}\cap\mathcal{E}_{1j}^c) +P(\tilde{\mathcal{E}}_{4j})+P(\tilde{\mathcal{E}}_{5j})+P(\mathcal{E}_{6j})+P(\mathcal{E}_{7j}). \nonumber
\end{align}
Similar to the proof of Theorem \ref{rate_two_way_nobin} in Appendix \ref{appendix:rate_two_way_nobin},
by the covering lemma, $P(\mathcal{E}_{1j})\rightarrow 0$ as $n\rightarrow \infty$, if
\begin{align}
\label{rate_1_1}
R_r > I(\hat{Y}_r;Y_r|X_r).
\end{align}
By the conditional typicality lemma, the second term $P(\mathcal{E}_{2j}\cap\mathcal{E}_{1j}^c) \rightarrow 0$ as $n\rightarrow \infty$. \par
By the packing lemma, $P(\tilde{\mathcal{E}}_{4j})\rightarrow 0$ as $n\rightarrow \infty$ if
\begin{align}
R_2 \leq I(X_2;Y_1,\hat{Y}_r|X_r,X_1).
\end{align}
Similarly, $P(\tilde{\mathcal{E}}_{5j})\rightarrow 0$ as $n\rightarrow \infty$ if
\begin{align}
\!\! R_2+R_r \leq I(X_2,X_r;Y_1|X_1)+I(\hat{Y}_r;X_1,X_2,Y_1|X_r).
\end{align}

For the new error events $\mathcal{E}_{6j}$ and $\mathcal{E}_{7j}$, the decoded joint distributions are as follows.
\begin{align}
\mathcal{E}'_{2j}(m_{2,j},1,k_{j-1})&: p(x_1)p(x_2)p(x_r)p(\hat{y}_r|x_r)p(y_1|x_1)\nonumber\\
\!\!\!\! \mathcal{E}'_{2j}(m_{2,j},k_j,k_{j-1})&: p(x_1)p(x_2)p(x_r)p(\hat{y}_r|x_r)p(y_1|x_1),
\label{joint_distribution_2}
\end{align}
where $m_{2,j} \neq 1, k_j \neq 1,k_{j-1} \neq 1$. Using standard joint typicality analysis, we can obtain a bound on each error event as follows.
$P(\mathcal{E}_{6j})\rightarrow 0$ as $n\rightarrow \infty$ if
\begin{align}
\!\! R_2+R_r \leq I(X_2,X_r;Y_1|X_1)+I(\hat{Y}_r;X_1,X_2,Y_1|X_r).
\end{align}
$P(\mathcal{E}_{7j})\rightarrow 0$ as $n\rightarrow \infty$ if
\begin{align}
R_2+2R_r \leq I(X_r;Y_1|X_1)+I(X_2,X_r;Y_1|X_1)
+I(\hat{Y}_r;X_1,X_2,Y_1|X_r).
\label{rate:new_no_repetitiion}
\end{align}
Combining the above inequalities, we obtain the rate region in Corollary \ref{cor:no_index_decoding}. Compared to the rate region of noisy network coding in \eqref{rate_noisy}, the only new constraint is \eqref{rate:new_no_repetitiion}, which occurs when both compression indices are wrong in addition to a wrong message.

\setcounter{subsubsection}{0}
\section{Proof of Corollary \ref{cor:repeat_once_decoding}}
\label{appendix:cor:repeat_once_decoding}
Assume without loss of generality that $m_{1,j}=m_{1,j+1}=m_{2,j}=1$ and $k_{2j-2}=k_{2j-1}=k_{2j}=1$. First define the following three events:
\begin{align}
\mathcal{E}'_{1j}(k_{2j})
&=\big\{(x_r^n(k_{2j}),y_1^n(2j+1),x_{1,2j+1}^n(1))\in T^{(n)}_{\epsilon}\big\}\nonumber\\
\mathcal{E}'_{2j}(m_{2,j},k_{2j},k_{2j-1})
&=\big\{(x_{2,2j}^n({m}_{2,j}),x_r^n(k_{2j-1}),\hat{y}_r^n({k}_{2j}|k_{2j-1}),y_1^n(2j),x_{1,2j}^n(1))\in T^{(n)}_{\epsilon}\big\} \nonumber\\
\mathcal{E}'_{3j}(m_{2,j},k_{2j-1},k_{2j-2})
&=\big\{(x_{2,2j-1}^n({m}_{2,j}),x_r^n(k_{2j-2}),\hat{y}_r^n({k}_{2j-1}|k_{2j-2}),y_1^n(2j-1),x_{1,2j-1}^n(1))\in T^{(n)}_{\epsilon}\big\}. \nonumber
\end{align}
Then the decoder makes an error only if one or more of the following events occur:
\begin{align}
\mathcal{E}_{1j}&=\big\{(\hat{y}_r^n(k_{j}|1),y^n_r(j),x^n_r(1))\notin T^{(n)}_{\epsilon}~\textrm{for all}~k_j\in[1:2^{nR_r}]\big\} \nonumber \\
\mathcal{E}_{2j}&=\big\{(x_r^n(1),y_1^n(2j+1),x_{1,2j+1}^n(1))\notin T^{(n)}_{\epsilon}~\textrm{or}~\nonumber \\
&~~~~~\;(x_{2,2j}^n(1),x_r^n(1),\hat{y}_r^n(1|1),y_1^n(2j),x_{1,2j}^n(1))\notin T^{(n)}_{\epsilon} ~\textrm{or}~\nonumber \\
&~~~~~\;(x_{2,2j-1}^n(1),x_r^n(1),\hat{y}_r^n(1|1),y_1^n(2j-1),x_{1,2j-1}^n(1))\notin T^{(n)}_{\epsilon}\big\} \nonumber\\
\mathcal{E}^r_{3j}&=\big\{\mathcal{E}'_{1j}(1), \mathcal{E}'_{2j}(m_{2,j},1,1)~\textrm{and}~\mathcal{E}'_{3j}(m_{2,j},1,1)~\textrm{for some}~m_{2,j}\neq 1\big\} \nonumber \\
\mathcal{E}^r_{4j}&=\big\{\mathcal{E}'_{1j}(1), \mathcal{E}'_{2j}(m_{2,j},1,1)~\textrm{and}~\mathcal{E}'_{3j}(m_{2,j},1,k_{2j-2})~\textrm{for some}~m_{2,j}\neq 1,k_{2j-2}\neq 1\big\} \nonumber \\
\mathcal{E}^r_{5j}&=\big\{\mathcal{E}'_{1j}(1), \mathcal{E}'_{2j}(m_{2,j},1,k_{2j-1})~\textrm{and}~\mathcal{E}'_{3j}(m_{2,j},k_{2j-1},1)~\textrm{for some}~m_{2,j}\neq 1,k_{2j-1}\neq 1\big\} \nonumber \\
\mathcal{E}^r_{6j}&=\big\{\mathcal{E}'_{1j}(k_{2j}), \mathcal{E}'_{2j}(m_{2,j},k_{2j},1)~\textrm{and}~\mathcal{E}'_{3j}(m_{2,j},1,1)~\textrm{for some}~m_{2,j}\neq 1,k_{2j}\neq 1\big\}  \nonumber
\end{align}
\begin{align}
\mathcal{E}^r_{7j}&=\big\{\mathcal{E}'_{1j}(k_{2j}), \mathcal{E}'_{2j}(m_{2,j},k_{2j},k_{2j-1})~\textrm{and}~\mathcal{E}'_{3j}(m_{2,j},k_{2j-1},1)~\textrm{for some}~m_{2,j}\neq 1,k_{2j}\neq 1,k_{2j-1}\neq 1\big\} \nonumber \\
\mathcal{E}^r_{8j}&=\big\{\mathcal{E}'_{1j}(k_{2j}), \mathcal{E}'_{2j}(m_{2,j},k_{2j},1)~\textrm{and}~\mathcal{E}'_{3j}(m_{2,j},1,k_{2j-2})~\textrm{for some}~m_{2,j}\neq 1,k_{2j}\neq 1,k_{2j-2}\neq 1\big\} \nonumber \\
\mathcal{E}^r_{9j}&=\big\{\mathcal{E}'_{1j}(1), \mathcal{E}'_{2j}(m_{2,j},1,k_{2j-1})~\textrm{and}~\mathcal{E}'_{3j}(m_{2,j},k_{2j-1},k_{2j-2})~\textrm{for some}~m_{2,j}\neq 1,k_{2j-1}\neq 1,k_{2j-2}\neq 1\big\} \nonumber \\
\mathcal{E}^r_{10j}&=\big\{\mathcal{E}'_{1j}(k_{2j}), \mathcal{E}'_{2j}(m_{2,j},k_{2j},k_{2j-1})~\textrm{and}~\mathcal{E}'_{3j}(m_{2,j},k_{2j-1},k_{2j-2})\nonumber\\
&~~~~~~\textrm{for some}~m_{2,j}\neq 1,k_{2j}\neq 1,k_{2j-1}\neq 1,k_{2j-2}\neq 1\big\} \nonumber
\end{align}
Similar to the proof of Theorem \ref{rate_two_way_nobin} in Appendix \ref{appendix:rate_two_way_nobin},
by the covering lemma, $P(\mathcal{E}_{1j})\rightarrow 0$ as $n\rightarrow \infty$, if
\begin{align}
\label{rate_1_2}
R_r > I(\hat{Y}_r;Y_r|X_r).
\end{align}
By the conditional typicality lemma, the second term $P(\mathcal{E}_{2j}\cap\mathcal{E}_{1j}^c) \rightarrow 0$ as $n\rightarrow \infty$. \par
Let symbol "$*$" represent the wrong message or compression index.  For $l\in\{2j,3j\}$, similar to \eqref{joint_distribution_1} and \eqref{joint_distribution_2}, the joint decoded distributions for the rest of the error events are as follows.
\begin{align*}
\mathcal{E}'_{1j}(*)&:p(x_1)p(x_r)p(y_1|x_1)\\
\mathcal{E}'_{l}(*,1,1)&:p(x_1)p(x_2)p(x_r)p(y_1,\hat{y}_r|x_r,x_1)\\
\mathcal{E}'_{l}(*,1,*)&:p(x_1)p(x_2)p(x_r)p(\hat{y}_r|x_r)p(y_1|x_1)\\
\mathcal{E}'_{l}(*,*,1)&:p(x_1)p(x_2)p(x_r)p(\hat{y}_r|x_r)p(y_1|x_r,x_1)\\
\mathcal{E}'_{l}(*,*,*)&:p(x_1)p(x_2)p(x_r)p(\hat{y}_r|x_r)p(y_1|x_1).
\end{align*}
Using standard joint typicality analysis, we can obtain a bound on each error event as follows.\\
$P(\mathcal{E}^r_{3j})\rightarrow 0$ as $n\rightarrow \infty$ if
\begin{align}
2R_2 \leq 2I(X_2;Y_1,\hat{Y}_r|X_r,X_1).
\end{align}
$P(\mathcal{E}^r_{4j})\rightarrow 0$ as $n\rightarrow \infty$ if
\begin{align}
2R_2+R_r \leq\; I(X_2;Y_1,\hat{Y}_r|X_r,X_1)+I(X_2,X_r;Y_1|X_1)+I(\hat{Y}_r;X_1,X_2,Y_1|X_r).
\end{align}
$P(\mathcal{E}^r_{5j})\rightarrow 0$ as $n\rightarrow \infty$ if
\begin{align}
\!\!2R_2+R_r \leq\; I(X_2,X_r;Y_1|X_1)+I(\hat{Y}_r;X_1,X_2,Y_1|X_r)+I(X_2;Y_1|X_1,X_r)+I(\hat{Y}_r;X_1,X_2,Y_1|X_r).
\end{align}
$P(\mathcal{E}^r_{6j})\rightarrow 0$ as $n\rightarrow \infty$ if
\begin{align}
2R_2+R_r \leq\; I(X_2;Y_1,\hat{Y}_r|X_r,X_1)+I(X_2,X_r;Y_1|X_1)+I(\hat{Y}_r;X_1,X_2,Y_1|X_r).
\end{align}
$P(\mathcal{E}^r_{7j})\rightarrow 0$ as $n\rightarrow \infty$ if
\begin{align}
2R_2+2R_r \leq\;& I(X_r;Y_1|X_1)+I(X_2,X_r;Y_1|X_1)\nonumber\\
&+I(\hat{Y}_r;X_1,X_2,Y_1|X_r)+I(X_2;Y_1|X_1,X_r)+I(\hat{Y}_r;X_1,X_2,Y_1|X_r).
\end{align}
$P(\mathcal{E}^r_{8j})\rightarrow 0$ as $n\rightarrow \infty$ if
\begin{align}
2R_2+2R_r \leq 2[I(X_2,X_r;Y_1|X_1)+I(\hat{Y}_r;X_1,X_2,Y_1|X_r)].
\end{align}
$P(\mathcal{E}^r_{9j})\rightarrow 0$ as $n\rightarrow \infty$ if
\begin{align}
2R_2+2R_r \leq 2[I(X_2,X_r;Y_1|X_1)+I(\hat{Y}_r;X_1,X_2,Y_1|X_r)].
\end{align}
$P(\mathcal{E}^r_{10j})\rightarrow 0$ as $n\rightarrow \infty$ if
\begin{align}
\label{rate_10j}
2R_2+3R_r \leq\; I(X_r;Y_1|X_1)+2[I(X_2,X_r;Y_1|X_1)+I(\hat{Y}_r;X_1,X_2,Y_1|X_r)].
\end{align}
Combining the above inequalities, we obtain the rate region in Corollary \ref{cor:repeat_once_decoding}. The extra rate constraints \eqref{repeat_two_more_1} and \eqref{repeat_two_more_2} come only from combing \eqref{rate_10j} with \eqref{rate_1_2}. Again we see the boundary effect when the message and all compression indices are wrong.

\section{Proofs of Corollary \ref{cor:opti_sum_noisy} and Corollary \ref{sumrate_allequ}}
\subsection{Proof of Corollary \ref{cor:opti_sum_noisy}}
\label{appendix:cor:opti_sum_noisy}
Our objective is to find the optimal $\sigma^2$ which maximizes the sum rate
\begin{align}
s(\sigma^2) \triangleq \max\left\{0,\min\{R_{11}(\sigma^2),R_{12}(\sigma^2)\}\right\}+\max\left\{0,\min\{R_{21}(\sigma^2),R_{22}(\sigma^2)\}\right\}. \nonumber
\end{align}
Recall that we have assumed $\sigma_{e1}^2 \geq \sigma_{e2}^2$, where $\sigma_{e1}^2,\sigma_{e2}^2$ are defined in (\ref{e1e2}). Note that both $R_{11}(\sigma^2),R_{21}(\sigma^2)$ are non-increasing and $R_{12}(\sigma^2),R_{22}(\sigma^2)$ are non-decreasing. Also,
\begin{align}
R_{11}(\sigma_{e1}^2)&=R_{12}(\sigma_{e1}^2)\nonumber\\
R_{21}(\sigma_{e2}^2)&=R_{22}(\sigma_{e2}^2).\nonumber
\end{align}
Therefore, the problem of maximizing $s(\sigma^2)$ is equivalent to
\begin{align}
\max~~~&q(\sigma^2) \triangleq \max \{0,R_{12}(\sigma^2)\}+R_{21}(\sigma^2)\nonumber\\
\textrm{s.t.}~~~&\sigma^2 \leq \sigma_{e1}^2,\nonumber\\
&\sigma^2 \geq \sigma_{e2}^2.
\label{opti_prob}
\end{align}
Note that $R_{12}(\sigma_{z1}^2)=0$ and $R_{12}(\sigma^2)<0$ for $\sigma^2 \in (0,\sigma_{z1}^2)$, $R_{12}(\sigma^2)>0$ for $\sigma^2 \in (\sigma_{z1}^2,\infty)$. Therefore, the optimization problem in (\ref{opti_prob}) can be derived into the following two cases:\par
\textit{Case 1}: $\sigma_{z1}^2 \leq \sigma_{e2}^2$\\
In this case, the objective function $q(\sigma^2)$ can be simplified to $q(\sigma^2)=R_{12}(\sigma^2)+R_{21}(\sigma^2)$, which is continuously differentiable for $\sigma^2 \in [\sigma_{e2}^2,\sigma_{e1}^2]$. Thus the optimization problem is equivalent to
\begin{align*}
\min~~~&f(\sigma^2)=-R_{12}(\sigma^2)-R_{21}(\sigma^2)\\
\textrm{s.t.}~~~&c_1(\sigma^2)=\sigma_{e1}^2-\sigma^2 \geq 0,\\
&c_2(\sigma^2)=\sigma^2-\sigma_{e2}^2 \geq 0.
\end{align*}
Form the Lagrangian as  $\mathcal{L}(\sigma^2,\lambda_1,\lambda_2)=f(\sigma^2)-\lambda_1c_1(\sigma^2)-\lambda_2c_2(\sigma^2)$. With the KKT conditions, we have:
\begin{align}
\nabla_{\sigma^2}\mathcal{L}(\sigma_{N1}^2,\lambda_1,\lambda_2)&=0,\nonumber\\
\sigma_{e1}^2-\sigma_{N1}^2 &\geq 0,\nonumber\\
\sigma_{N1}^2-\sigma_{e2}^2 &\geq 0,\nonumber\\
\lambda_1,\lambda_2 &\geq 0,\nonumber\\
\lambda_1c_1(\sigma_{N1}^2)&=0,\nonumber\\
\lambda_2c_2(\sigma_{N1}^2)&=0.
\end{align}
By solving the above conditions, the optimal $\sigma_{N1}^2$ for this case is characterized as in (\ref{sigma_n1}).\par
\textit{Case 2}: $\sigma_{z1}^2 > \sigma_{e2}^2$\\
In this case, the objective function $q(\sigma^2)$ is no longer continuously differentiable for $\sigma^2 \in [\sigma_{e2}^2,\sigma_{e1}^2]$. To solve this problem, we divide this interval into two parts. When $\sigma^2 \in [\sigma_{e2}^2,\sigma_{z1}^2)$, the objective function is simplified to $q(\sigma^2)=R_{21}(\sigma^2)$ and is maximized at $\sigma^2=\sigma_{e2}^2$. When $\sigma^2 \in [\sigma_{z1}^2,\sigma_{e1}^2]$, the objective function is simplified to $q(\sigma^2)=R_{12}(\sigma^2)+R_{21}(\sigma^2)$ and is continuously differentiable for $\sigma^2 \in [\sigma_{z1}^2,\sigma_{e1}^2]$. With the KKT conditions, the optimal $\sigma^2$ for this case is characterized as $\sigma_{N2}^2$ in (\ref{sigma_n2}). Combining the two intervals, the optimal $\sigma_N^2$ for this case can be characterized as in (\ref{sigma_on2}).
\subsection{Proof of Corollary \ref{sumrate_allequ}}\label{appendix:cor:sumrate_allequ}

From Corollary \ref{cor:opti_sum_noisy} and Theorem \ref{same_sum_rate}, we have:\\
If $\sigma_{z1}^2 \leq \sigma_{e2}^2$, then the two schemes achieve the same sum rate if and only if at least one of the following three conditions holds:
\begin{align}
\sigma_{c1}^2 &\leq \sigma_{e2}^2\\
\label{three_two}
0<\sigma_{c1}^2 &\leq \sigma_{g}^2\\
\label{three_three}
\sigma_{g}^2 &\leq 0
\end{align}
For the channel configuration of $g_{r1}=g_{1r}, g_{r2}=g_{2r}, g_{21}=g_{12}$, we can show that either (\ref{three_two}) or (\ref{three_three}) will hold.\par
If $\sigma_{z1}^2 > \sigma_{e2}^2$, we can show that $R_{21}(\sigma_{e2}^2) < R_{12}(\sigma_{e1}^2)+R_{21}(\sigma_{e1}^2)$ always holds. Therefore, according to (\ref{sigma_n2}) and (\ref{sigma_on2}), the optimal $\sigma_{N}^2$ can be characterized as:
\begin{align*}
&\text{if}~\sigma_{e1}^2 \geq \sigma_{g}^2 \geq 0,~\sigma_{N}^2=\sigma_{g}^2;\\
&\text{if}~\sigma_{g}^2 > \sigma_{e1}^2~\text{or}~\sigma_{g}^2 < 0,~\sigma_{N}^2=\sigma_{e1}^2.
\end{align*}
Since either (\ref{three_two}) or (\ref{three_three}) holds and $\sigma_{c1}^2 \leq \sigma_{e1}^2$, we obtain $\sigma_{c1}^2 \leq \sigma_{N}^2$. According to Theorem \ref{same_sum_rate}, the two schemes then achieve the same sum rate.

\section{Proof of Theorem \ref{thm:fading_region}}\label{pro8}
From the assumption of CSI available at each node, the constraint $R_{11}$ in Theorem \ref{rate_two_way_nobin} can be expressed as
\begin{align}\label{dmc1}
R_{11}\leq&\; I(X_1;\tilde{Y}_2,\hat{Y}_r|X_2,X_r)\nonumber\\
=&\;I(X_1;Y_2,h_{21},h_{2r},h_{r1},h_{r2},\hat{Y}_r|X_2,X_r)\nonumber\\
\stackrel{(a)}{=}&\;I(X_1;Y_2,\hat{Y}_r|X_2,X_r,h_{21},h_{2r},h_{r1},h_{r2})\nonumber\\
=&\;h(Y_2,\hat{Y}_r|X_2,X_r,h_{21},h_{2r},h_{r1},h_{r2})-\;h(Y_2,\hat{Y}_r|X_1,X_2,X_r,h_{21},h_{2r},h_{r1},h_{r2})\nonumber\\
=&\;h \left( \left(\frac{h_{21}}{d_{12}^{\alpha/2}}X_1+Z_2\right),
\left(\frac{h_{r1}}{d_{1r}^{\alpha/2}}X_1+Z_r+\hat{Z}\right) \Bigg| h_{r_1},h_{21}\right)-\;h(Z_2,Z_r+\hat{Z})
\end{align}
\noindent where $(a)$ follows since each $h_{ij}$ is independent from all $X_i$. Similarly, we obtain $R_{21}$ as
\begin{align}\label{dmc2}
R_{21}\leq\;h\left( \left(\frac{h_{12}}{d_{12}^{\alpha/2}}X_2+Z_1\right),
\left(\frac{h_{r2}}{d_{2r}^{\alpha/2}}X_2+Z_r+\hat{Z}\right)\Bigg|h_{r_2},h_{12}\right)-\;h(Z_1,Z_r+\hat{Z}).
\end{align}
Moving to $R_{12}$, we can express it as
\begin{align}\label{dmc3}
R_{12}
\leq&\;I(X_1,X_r;\tilde{Y}_2|X_2)-I(\hat{Y_r};\tilde{Y}_r|X_1,X_2,X_r,\tilde{Y}_2)\nonumber\\
=&\;I(X_1,X_r;Y_2,h_{21},h_{2r},h_{r1},h_{r2}|X_2)-I(\hat{Y_r};Y_r,h_{r1},h_{r2}|X_1,X_2,X_r,Y_2,h_{21},h_{2r},h_{r1},h_{r2})\nonumber\\
=&\;I(X_1,X_r;Y_2|X_2,h_{21},h_{2r},h_{r1},h_{r2})-I(\hat{Y_r};Y_r,h_{r1},h_{r2}|X_1,X_2,X_r,Y_2,h_{21},h_{2r},h_{r1},h_{r2})\nonumber\\
=&\;h(Y_2|X_2,h_{21},h_{2r})-h(Y_2|X_1,X_r,X_2,h_{21},h_{2r})-h(\hat{Y_r}|X_1,X_2,X_r,Y_2,h_{21},h_{2r},h_{r1},h_{r2})\nonumber\\
&+I(\hat{Y_r}|X_1,X_2,X_r,Y_2,h_{21},h_{2r},h_{r1},h_{r2},Y_r)\nonumber\\
=&\;h\left(\frac{h_{21}}{d_{12}^{\alpha/2}}X_1+\frac{h_{2r}}{d_{2r}^{\alpha/2}}X_r+Z_2\Bigg|h_{21},h_{2r}\right)-h(Z_2)-h(Z_r+\hat{Z})+h(\hat{Z}).
\end{align}
Similarly, we obtain $R_{22}$ as
\begin{align}\label{dmc4}
R_{22}\leq&\;h\left(\frac{h_{12}}{d_{12}^{\alpha/2}}X_2+\frac{h_{1r}}{d_{1r}^{\alpha/2}}X_r+Z_1\Bigg|h_{12},h_{1r}\right)-h(Z_1)-h(Z_r+\hat{Z})+h(\hat{Z}).
\end{align}
Moving to the compression rate constraints, we obtain
\begin{align}\label{dmc5}
&I(\hat{Y_r};\tilde{Y}_r|X_1,X_2,X_r,\tilde{Y}_1)\leq\;I(X_r;\tilde{Y}_1|X_1)\nonumber\\
&\Leftrightarrow h(Z_r+\hat{Z})-h(\hat{Z})\leq\;I(X_r;Y_1|X_1,h_{12},h_{1r},h_{r1},h_{r2})\nonumber\\
&\Leftrightarrow h(Z_r+\hat{Z})-h(\hat{Z})\leq h\left(\frac{h_{12}}{d_{12}^{\alpha/2}}X_2+\frac{h_{1r}}{d_{1r}^{\alpha/2}}X_r+Z_1\Bigg|h_{12},h_{1r}\right)-h\left(\frac{h_{12}}{d_{12}^{\alpha/2}}X_2+Z_1\Bigg|h_{12}\right).
\end{align}
Similarly, the second constraint is equivalent to
\begin{align}\label{dmc6}
&I(\hat{Y_r};\tilde{Y}_r|X_1,X_2,X_r,\tilde{Y}_2)\leq\;I(X_r;\tilde{Y}_2|X_2)\nonumber\\
&\Leftrightarrow h(Z_r+\hat{Z})-h(\hat{Z})\leq h\left(\frac{h_{21}}{d_{12}^{\alpha/2}}X_1+\frac{h_{2r}}{d_{2r}^{\alpha/2}}X_r+Z_2\Bigg|h_{21},h_{2r}\right)-h\left(\frac{h_{21}}{d_{12}^{\alpha/2}}X_1+Z_2\Bigg|h_{21}\right).
\end{align}
Regarding the optimal input distributions, although Gaussian distributions are not necessarily optimal for the compress-forward strategy, we choose $X_1$, $X_2$, $X_r$ and $\hat{Z}$ to be Gaussian for simplicity. With this assumption, equations (\ref{dmc1}-\ref{dmc4}) become as those in (\ref{fad1}) \cite{gamal2010lecture}. Hence, we obtain the rate constraints in (\ref{fad5}). Moreover, constraints (\ref{dmc5}, \ref{dmc6}) become
\begin{align}\label{fad10}
C(1/\sigma^2)\leq&\;E\left\{C\left(\frac{\frac{|h_{2r}|^2}{d_{2r}^\alpha}P}{1+\frac{|h_{21}|^2}{d_{12}^\alpha}P}\right)\right\}\nonumber\\
C(1/\sigma^2)\leq&\;E\left\{C\left(\frac{\frac{|h_{1r}|^2}{d_{1r}^\alpha}P}{1+\frac{|h_{12}|^2}{d_{12}^\alpha}P}\right)\right\}.
\end{align}
\noindent From (\ref{fad10}), we obtain the constraints in (\ref{consre}).

\section{Proof of Corollary \ref{cor:fading_express}}\label{pro9}
Since $|h_{12}|$ is Rayleigh, then $|h_{12}|^2$ is exponential with mean $\lambda=1/E[|h_{12}|^2]$. Let $\gamma=u+v$, where $u$ and $v$ are independent exponential random variables with
parameters $\lambda_u$ and $\lambda_v$, respectively. Then the CDF of $\gamma$ is given as \cite{Lantse}
\begin{align*}
P(\gamma)=
\left\{\begin{array}{cc}
                  1-\left[\frac{\lambda_v}{\lambda_v-\lambda_u}e^{-\lambda_u\gamma}+
                  \frac{\lambda_u}{\lambda_u-\lambda_v}e^{-\lambda_v\gamma}\right], & \lambda_u\neq \lambda_v \\
                  1-(1+\lambda\gamma)e^{\lambda\gamma}, & \lambda_u= \lambda_v
                \end{array}\right\}.
\end{align*}
\noindent Then, the pdf of $\gamma$ for $\lambda_u\neq \lambda_v$ can be expressed as
\begin{align}\label{fad3}
p(\gamma)=\frac{\lambda_u\lambda_v}{\lambda_v-\lambda_u}\left[e^{-\lambda_u\gamma}-e^{-\lambda_v\gamma}\right].
\end{align}
In \cite{sagwei}, it is shown that
\begin{align}\label{fad4}
\int_0^{\infty}\log(1+\gamma)\cdot \frac{1}{\bar{\gamma}}e^{-\frac{\gamma}{\bar{\gamma}}}d\gamma
=\frac{1}{4\pi\sqrt{\pi}\ln{2}\bar{\gamma}}G_{4,6}^{6,2}\left[\frac{1}{(2\bar{\gamma})^2}\mid_{0,0.5,-0.5,0,-0.5,0}^{-0.5,0,0,0.5}\right].
\end{align}
To evaluate $\bar{R}_{11}$ in (\ref{fad1}), set $u_1=\frac{|h_{21}|^2P}{d_{12}^\alpha}$, $v_1=\frac{|h_{r1}|^2P}{d_{1r}^\alpha(1+\sigma^2)}$ and $\gamma=u_1+v_1$. Then, by using (\ref{fad3}) and (\ref{fad4}),  $\bar{R}_{11}$ can be expressed as
\begin{align*}
\bar{R}_{11}=\!\!\int_0^{\infty}\!\!\!\log(1+\gamma)\cdot
\frac{\lambda_{u_1}\lambda_{v_1}}{\lambda_{v_1}-\lambda_{u_1}}\left[e^{-\lambda_{u_1}\gamma}-e^{-\lambda_{v_1}\gamma}\right]d\gamma
\end{align*}
\noindent which can be solved as in (\ref{fad2}). Similar derivation holds for the other rates.

\section{Evaluate the expressions in Corollary \ref{cor:fading_sumrate}}\label{pro10}
Staring with $\bar{\sigma}_{z1}^2$, we use the second equation in (\ref{fad2}) to obtain the same $\bar{\sigma}_{z1}^2$ in (\ref{sigz}) with
\begin{align}
D_1=\;\frac{\lambda_{u_1}\lambda_{v_3}}{4\pi\sqrt{\pi}\ln{2}(\lambda_{v_3}-\lambda_{u_1})}
\times\left(G_{4,6}^{6,2}
\left[\left(\frac{\lambda_{u_1}}{2}\right)^2\Bigg|_{b_q}^{a_p}\right]
-G_{4,6}^{6,2}\left[\left(\frac{\lambda_{v_3}}{2}\right)^2\Bigg|_{b_q}^{a_p}\right]\right).
\end{align}
Next, the maximum sum rate is obtained by solving the following problem:
\begin{align*}
\max_{\substack{\sigma}}\;\bar{R}_{12}&+\bar{R}_{21}\nonumber\\
\text{s.t.}\;c_1=\bar{\sigma}_{e1}^2-\sigma^2&\geq\;0,\nonumber\\
c_2=\sigma^2-\bar{\sigma}_{e2}^2&\geq\;0.
\end{align*}
Since the constraints $c_1$ and $c_2$ are inactive for the interval $(\bar{\sigma}_{e2}^2,\bar{\sigma}_{e1}^2)$, their corresponding Lagrange multipliers are equal to zero and the optimal $\sigma$ is obtained by taking the derivative of $\bar{R}_{12}+\bar{R}_{21}$ with respect to $\sigma$, which can be expressed as
\begin{align}\label{derf}
\frac{\partial}{\partial \sigma}(\bar{R}_{12}+\bar{R}_{21})=\frac{\partial \bar{R}_{12}}{\partial \bar{\sigma}}+\frac{\partial \bar{R}_{21}}{\partial \bar{\sigma}}
=\frac{2}{\bar{\sigma}(1+\bar{\sigma}^2)}+\frac{\partial \bar{R}_{21}}{\partial \bar{\sigma}}.
\end{align}
\noindent Since
\begin{align}\label{ders}
\frac{\partial \bar{R}_{21}}{\partial \bar{\sigma}}=\frac{\partial \bar{R_{21}}}{\partial \lambda_{v_2}}\frac{\partial \lambda_{v_2}}{\partial \bar{\sigma}}=\frac{2d_{2r}^{\alpha}\bar{\sigma} }{E[|g_{r_2}|^2]P}\frac{\partial \bar{R}_{21}}{\partial \lambda_{v_2}},
\end{align}
\noindent we can express the rate $\bar{R}_{21}$ as $\acute{R}_{21}+\grave{R}_{21}$ where
\begin{align}
\acute{R}_{21}=&\;\frac{\lambda_{u_2}\lambda_{v_2}}{4\pi\sqrt{\pi}\ln{2}(\lambda_{v_2}-\lambda_{u_2})}
G_{4,6}^{6,2}\left[\left(\frac{\lambda_{u_2}}{2}\right)^2\Bigg|_{b_q}^{a_p}\right],\;\text{and}\nonumber\\
\grave{R}_{21}=&\;-\frac{\lambda_{u_2}\lambda_{v_2}}{4\pi\sqrt{\pi}\ln{2}(\lambda_{v_2}-\lambda_{u_2})}
G_{4,6}^{6,2}\left[\left(\frac{\lambda_{v_2}}{2}\right)^2\Bigg|_{b_q}^{a_p}\right]\nonumber\\
=&\;-\frac{\lambda_{u_2}}{4\pi\sqrt{\pi}\ln{2}}\frac{2}{(\lambda_{v_2}-\lambda_{u_2})}
\times\;((0.5\lambda_{v_2})^2)^{1-0.5}
G_{4,6}^{6,2}\left[\left(\frac{\lambda_{v_2}}{2}\right)^2\Bigg|_{b_q}^{a_p}\right]\nonumber\\
=&\;-\frac{\lambda_{u_2}}{4\pi\sqrt{\pi}\ln{2}}\frac{2}{(\lambda_{v_2}-\lambda_{u_2})}(Q^{1-0.5})
G_{4,6}^{6,2}\left[Q\mid_{b_q}^{a_p}\right]\nonumber\\
=&\;W\cdot F_1\cdot F_2
\end{align}
\noindent with
\begin{align*}
W=&\;-\frac{\lambda_{u_2}}{4\pi\sqrt{\pi}\ln{2}}\nonumber\\
F_1=&\;\frac{2}{(\lambda_{v_2}-\lambda_{u_2})}\nonumber\\
F_2=&\;(Q^{1-0.5})
G_{4,6}^{6,2}\left[Q\mid_{b_q}^{a_p}\right].
\end{align*}
\noindent Then, we have
\begin{align}\label{awr12}
\!\!\frac{\partial \acute{R}_{21}}{\partial \lambda_{v_2}}=&\;\frac{-\lambda_{u_2}^2}{(4\pi\sqrt{\pi}\ln{2})(\lambda_{v_2}-\lambda_{u_2})^2}
G_{4,6}^{6,2}\left[\left(\frac{\lambda_{u_2}}{2}\right)^2\Bigg|_{b_q}^{a_p}\right]
\end{align}
\noindent and the derivative of $\grave{R}_{21}$ can be expressed as
\begin{align}\label{der1}
\frac{\partial \grave{R}_{21}}{\partial \lambda_{v_2}}=&\;W\left(F_1\cdot \frac{\partial F_2}{\partial Q}\cdot \frac{\partial Q}{\partial \lambda_{v_2}}+F_2\frac{\partial F_1}{\partial \lambda_{v_2}}\right)
\end{align}
\noindent where each term can be expressed as
\begin{align}\label{der2}
\frac{\partial F_2}{\partial Q}\stackrel{(a)}{=}&\;-Q^{-0.5}G_{4,6}^{6,2}\left[Q\mid_{b_q}^{-0.5,0,0,-0.5}\right]\nonumber\\
\frac{\partial Q}{\partial \lambda_{v_2}}=&\;\frac{\lambda_{v_2}}{2}\nonumber\\
\frac{\partial F_1}{\partial \lambda_{v_2}}=&\;\frac{-2}{(\lambda_{v_2}-\lambda_{u_2})^2}
\end{align}
\noindent with $(a)$ follows from the relationship
\begin{align}
\frac{\partial}{\partial z}\left[z^{1-a_p}G_{p,q}^{m,n}\left(z\mid_{\boldsymbol{b}_q}^{\boldsymbol{a}_p}\right)\right]=-z^{-a_p}G_{p,q}^{m,n}\left(z\mid_{\boldsymbol{b}_q}^{a_1,\cdot\cdot\cdot,a_{p-1},a_p-1}\right),\;n<p.
\end{align}
\noindent By substituting (\ref{der2}) into (\ref{der1}) and combining it with (\ref{awr12}), we obtain
\begin{align}\label{derl}
\frac{\partial \bar{R}_{21}}{\partial \lambda_{v_2}}=\;&\frac{\lambda_{u_2}}{(4\pi\sqrt{\pi}\ln{2})(\lambda_{v_2}-\lambda_{u_2})}\nonumber\\
&\Bigg(\frac{\lambda_{v_2}}{\lambda_{v_2}-\lambda_{u_2}}G_{4,6}^{6,2}\left[\left(\frac{\lambda_{v_2}}{2}\right)^2\Bigg|_{b_q}^{a_p}\right]
-\;\frac{\lambda_{u_2}}{\lambda_{v_2}-\lambda_{u_2}}G_{4,6}^{6,2}\left[\left(\frac{\lambda_{u_2}}{2}\right)^2\Bigg|_{b_q}^{a_p}\right]+2G_{4,6}^{6,2}\left[\left(\frac{\lambda_{v_2}}{2}\right)^2\Bigg|_{b_q}^{z_p}\right]\Bigg).
\end{align}
Finally, by substituting (\ref{derl}) into (\ref{ders}) and by setting (\ref{derf})$=0$, $\bar{\sigma}_g^2$ is obtained by solving the following equation:
\begin{align}
1=&\;\frac{d_{2r}^{\alpha}\bar{\sigma}_g^2(1+\bar{\sigma}_g^2)\lambda_{u_2}}{(4\pi\sqrt{\pi}\ln{2})E[|h_{r2}|^2]P(\lambda_{v_2}-\lambda_{u_2})}\nonumber\\
&\;\Bigg(\frac{-\lambda_{v_2}}{\lambda_{v_2}-\lambda_{u_2}}G_{4,6}^{6,2}\left[\left(\frac{\lambda_{v_2}}{2}\right)^2\Bigg|_{b_q}^{a_p}\right]+\;\frac{\lambda_{u_2}}{\lambda_{v_2}-\lambda_{u_2}}G_{4,6}^{6,2}\left[\left(\frac{\lambda_{u_2}}{2}\right)^2\Bigg|_{b_q}^{a_p}\right]-\;2G_{4,6}^{6,2}\left[\left(\frac{\lambda_{v_2}}{2}\right)^2\Bigg|_{b_q}^{z_p}\right]\Bigg)
\end{align}
\noindent where $a_p=[-0.5,0,0,0.5],$ $b_q=[0,0.5,-0.5,0,-0.5,0]$ and $z_p=[-0.5,0,0,-0.5].$

\bibliographystyle{IEEEtran}
\bibliography{reflist}

\end{document}